\documentclass[11pt]{article}
\usepackage{amsthm}
\usepackage{amssymb}
\usepackage{amsmath}
\usepackage{todonotes}
\usepackage[margin=1in]{geometry}
\usepackage[utf8]{inputenc}
\usepackage{hyperref}
\usepackage{mathptmx}
\usepackage{pbox}
\usepackage{framed}
\usepackage{mathtools}
\usepackage{empheq}
\usepackage{thmtools}
\usepackage{thm-restate}

\theoremstyle{plain}
\newtheorem{thm}{Theorem}
\numberwithin{thm}{section}

\newtheorem{prop}[thm]{Proposition}

\newtheorem{lem}[thm]{Lemma}

\theoremstyle{definition}
\newtheorem{df}{Definition}

\numberwithin{df}{section}

\theoremstyle{remark}
\newtheorem*{rem}{Remark}
\newcommand{\size}{\operatorname{size}}
  \newcommand{\Res}{\operatorname{Res}}

\newcommand{\PCSP}{\operatorname{PCSP}}
\newcommand{\CSP}{\operatorname{CSP}}

\newcommand{\Pol}{\operatorname{Pol}}
\newcommand{\poly}{\operatorname{poly}}
\newcommand{\Ham}{\operatorname{Ham}}

\newcommand{\AND}{\operatorname{AND}}
\newcommand{\OR}{\operatorname{OR}}
\newcommand{\PAR}{\operatorname{PAR}}
\newcommand{\quot}{\operatorname{quot}}
\newcommand{\MAJ}{\operatorname{MAJ}}
\newcommand{\MALT}{\operatorname{AT}}
\newcommand{\AT}{\MALT}

\newcommand{\NEQ}{\operatorname{NEQ}}
\newcommand{\id}{\operatorname{id}}

\newcommand{\THR}{\operatorname{THR}}
\newcommand{\PER}{\operatorname{PER}}
\newcommand{\Conv}{\operatorname{Conv}}
\newcommand{\Aff}{\operatorname{Aff}}
\newcommand{\THRPER}{\operatorname{THR-PER}}
\newcommand{\lcm}{\operatorname{lcm}}
\newcommand{\ar}{\operatorname{ar}}

\newcommand{\REG}{\operatorname{REG}}
\newcommand{\REGPER}{\operatorname{REG-PER}}
\newcommand{\LP}{\operatorname{LP}}
\newcommand{\LE}{\operatorname{LE}}
\newcommand{\LPLE}{\operatorname{LPLE}}

\newcommand{\ot}{\leftarrow}

\title{An Algorithmic Blend of LPs and Ring Equations for Promise CSPs}
\author{Joshua Brakensiek\thanks{Carnegie Mellon University, Pittsburgh, PA 15213, USA. Email: {\tt jbrakens@andrew.cmu.edu}. Research supported in part by an REU supplement to NSF CCF-1526092.} \and Venkatesan Guruswami\thanks{Computer Science Department, Carnegie Mellon University, Pittsburgh, PA 15213. Email: {\tt venkatg@cs.cmu.edu}. Research supported in part by NSF grants CCF-1422045 and CCF-1526092.} }
\date{}
\begin{document}

\maketitle
\thispagestyle{empty}
\begin{abstract}
Promise CSPs are a relaxation of constraint satisfaction problems where the goal is to find an assignment satisfying a relaxed version of the constraints. Several well known problems can be cast as promise CSPs including approximate graph and hypergraph coloring, discrepancy minimization, and interesting variants of satisfiability. Similar to CSPs, the tractability of promise CSPs can be tied to the structure of associated  operations on the solution space called (weak) polymorphisms. However, compared to CSPs whose polymorphisms are well-structured algebraic objects called clones, polymorphisms in the promise world are much less constrained --- essentially any infinite family of functions obeying mild conditions can arise as polymorphisms. Under the thesis that non-trivial polymorphisms govern tractability, promise CSPs therefore provide a fertile ground for the discovery of novel algorithms.

\smallskip
In previous work, we classified all tractable cases of Boolean promise CSPs when the constraint predicates are symmetric. The algorithms were governed by three kinds of polymorphism families: (i) parity functions, (ii) majority functions, or (iii) a non-symmetric (albeit block-symmetric) family we called alternating threshold. In this work, we provide a vast generalization of these algorithmic results. Specifically, we show that promise CSPs that admit a family of ``regional-periodic" polymorphisms are solvable in polynomial time, assuming that determining which region a point is in can be computed in polynomial time. Such polymorphisms are quite general and are obtained by gluing together several functions that are periodic in the Hamming weights in different blocks of the input. For example, we can have functions that equal parity for relative Hamming weights up to 1/2, and Majority (so identically 1) for weights above 1/2. 

\smallskip
Our algorithm is based on a novel combination of linear programming and solving linear systems over rings. We also abstract a framework based on reducing a promise CSP to a CSP over an infinite domain, solving it there (via the said combination of LPs and ring equations), and then rounding the solution to an assignment for the promise CSP instance. The rounding step is intimately tied to the family of polymorphisms, and clarifies the connection between polymorphisms and algorithms in this context. As a key ingredient, we introduce the technique of finding a solution to a linear program with integer coefficients that lies in a different ring (such as $\mathbb Z[\sqrt{2}]$) to bypass ad-hoc adjustments for lying on a rounding boundary.
\end{abstract}
\newpage

\thispagestyle{empty}
\tableofcontents

\newpage

\section{Introduction}\label{sec:intro}

Constraint satisfaction problems (CSPs) have driven some of the most influential developments in theoretical computer science, 
from NP-completeness to the PCP theorem to semidefinite programming algorithms to the Unique Games conjecture. The (recently settled~\cite{DBLP:conf/focs/Bulatov17,DBLP:conf/focs/Zhuk17})
algebraic dichotomy conjecture~\cite{DBLP:journals/siamcomp/FederV98,Bulatov2005} 
establishes that all CSPs are either NP-complete or decidable in polynomial time. Further, this line of work pinpoints the mathematical structure that allows for efficient algorithms: when the
solution space admits certain non-trivial closure operations called \emph{polymorphisms}, the CSP is tractable, and
otherwise it is NP-hard. For instance, for linear equations, if $v_1,v_2,v_3$ are three solutions, then so is $v_1-v_2+v_3$,
and the underlying polymorphism is $f(x,y,z) = x-y+z$. 

Such polymorphisms and resulting CSP algorithms are unfortunately relatively rare. For instance, in the Boolean case, where the dichotomy has been long known~\cite{Schaefer:1978}, there are only three non-trivial tractable cases: Horn SAT (along with its complement dual Horn SAT), 2-CNF satisfiability, and Linear Equations mod 2. The situation for larger domains is similar, with even arity two CSPs like graph $k$-colorability being NP-hard for $k \ge 3$. One well-studied approach to cope with the prevalent intractability of CSPs is to settle for approximation algorithms that satisfy a guaranteed fraction of constraints (the \emph{Max CSP} problem). This has been a very fruitful avenue of research from both the algorithmic and hardness sides.
In this context, a general algorithm based on semidefinite programming is known to deliver approximation guarantees matching the performance of a variant of polymorphisms tailored to optimization (namely ``low-influence approximate polymorphisms'')~\cite{DBLP:journals/corr/Brown-CohenR15}, and the Unique Games conjecture implies this cannot be improved upon~\cite{DBLP:journals/siamcomp/KhotKMO07,prasad,RaghavendraThesis}. Thus, at least conjecturally, we have a link between mathematical structure and the existence of efficient approximation algorithms, although notably such work does not apply to the approximation of satisfiable CSP instances.

\subsection{Promise CSPs and Polymorphisms}

The Max CSP framework, however, does not capture problems like approximate graph coloring where one is allowed more colors than the chromatic number of the graph, for example $10$-coloring a $3$-colorable graph. An extension of CSPs, called promise CSPs, captures such problems. Informally, a promise CSP asks for an assignment to a CSP instance that satisfies a relaxed version of the CSP instance. For instance, given a $k$-SAT instance promised to have an assignment satisfying $3$ literals per clause, we might settle for an assignment satisfying an odd number of literals in each clause. (We will give formal and more general definitions in Section~\ref{sec:prelim}, but briefly a promise CSP is defined by pairs of predicates $(P_i,Q_i)$ with $P_i \subseteq Q_i$, and given an instance of CSP with defining predicates $\{P_i\}$, we would like to find an assignment that satisfies the instance when $P_i$ is replaced with $Q_i$.) A promise CSP called $(2+\epsilon)$-SAT (and a variant related to $2$-coloring low-discrepancy hypergraphs) was studied in \cite{DBLP:journals/siamcomp/AustrinGH17}. This work also brought to the fore the concept of polymorphisms\footnote{Previous literature (e.g., \cite{DBLP:journals/siamcomp/AustrinGH17}) called polymorphisms of promise CSPs ``weak'' polymorphisms. When confusion may arise with their CSP counterparts, we refer to them as ``polymorphisms of promise CSPs''} associated with the promise CSP, which are functions that are guaranteed to map tuples in $P_i$ into $Q_i$ for every $i$, generalizing the concept of polymorphisms from the case when $P_i=Q_i$ (again, see Section~\ref{sec:prelim} for formal definitions). Some new hardness results for graph and hypergraph coloring were then obtained using the polymorphism framework in \cite{DBLP:conf/coco/BrakensiekG16}, and was also used to settle the notorious 3 vs. 5 coloring hardness in \cite{BKO18}.

In \cite{DBLP:conf/soda/BrakensiekG18}, we undertook a systematic investigation of promise CSPs via the lens of polymorphisms, building some theory of their structure and interplay with both algorithms and complexity. For the latter, there is a Galois correspondence implying that the complexity of a promise CSP is completely dictated by its polymorphisms~\cite{Pippenger2002}. Thus, from the perspective of classifying the complexity of promise CSPs, one can just focus on polymorphisms and forget about the relations defining the CSP. 

Our work, however, revealed that the space of polymorphisms for promise CSPs is very rich. Therefore, the program of classifying the complexity of promise CSPs via polymorphisms (along the lines of the successful theory establishing a dichotomy in the case of CSPs) must overcome significant challenges that go well beyond the CSP case.
The polymorphisms associated with CSPs are closed under compositions (since the output belongs to the same relation as the inputs), and as a result they belong to a well-structured class of objects in universal algebra called clones. Polymorphisms for promise CSPs inherently lose this closure under composition (as the output no longer belongs to the same relation as the inputs). They are therefore much less constrained --- essentially any family of functions obeying mild conditions (projection-closed and finitizable) can arise as polymorphisms~\cite{DBLP:conf/soda/BrakensiekG18, Pippenger2002}. Further, whereas a single non-trivial polymorphism can suffice for tractability (as it can be composed with itself to give more complex and higher arity functions), in the case of polymorphisms we really need an \emph{infinite family} of them in order to develop algorithms for the associated promise CSP. Indeed, the hardness results of \cite{DBLP:journals/siamcomp/AustrinGH17,DBLP:conf/soda/BrakensiekG18} proceed by establishing a junta-like structure for the polymorphisms, and thus the lack of a rich infinite family of them. 

The vast variety of possible families of polymorphisms means that there are still numerous algorithms, and possibly whole new algorithmic paradigms, yet to be discovered in the promise CSP framework. This is the broad agenda driving this work. Our main result in \cite{DBLP:conf/soda/BrakensiekG18} classified all tractable cases of Boolean promise CSPs whose defining predicates ($P_i,Q_i$) are symmetric.\footnote{A predicate $P$ is symmetric if for all $(x_1, \hdots, x_m) \in P$ and all permutations $\pi : [m]\to [m]$, we have that $(x_{\pi(1)}, \hdots, x_{\pi(m)}) \in P.$ We say that ($P_i, Q_i$) is symmetric if both $P_i$ and $Q_i$ are symmetric.} The algorithms were governed by (essentially) three nicely structured polymorphism families: (i) parity functions, (ii) majority functions, or (iii) a non-symmetric (albeit block-symmetric) family we called alternating threshold (see Theorem~\ref{thm:BG18-main} for the precise statement). 

\subsection{Our results}

This work is motivated by the program of more systematically leveraging families of polymorphisms toward the development of new algorithmic approaches to promise CSPs. In this vein, we provide a vast generalization of the above-mentioned algorithmic results for the case of symmetric Boolean promise CSPs, by exhibiting algorithms based on rather general (albeit still structured) families of polymorphisms.\footnote{One possible concern is that no promise CSPs (or only ``trivial'' promise CSPs) admit such a family $\mathcal F$ of polymorphisms. 
  A detailed discussion of why this is not the case is available in Appendix~E of the full version of \cite{DBLP:conf/soda/BrakensiekG18}.} Specifically, we show that promise CSPs that admit a family of ``regional-periodic'' polymorphisms are polynomial time solvable. Such polymorphisms are quite general; their precise description is a bit technical but at a high level they are obtained by gluing together, for various ranges of Hamming weights in prescribed blocks of the input, functions that are periodic in the Hamming weights in their respective block.\footnote{While we focus on the case that the domain of the $P_i$'s is Boolean (although the $Q_i$'s can be over any finite domain), this is mostly for notational simplicity. Our methods are general enough to be readily adapted to any finite domain; see Section~\ref{sec:larger}.} 

The algorithms require a novel combination of linear programming and solving equations over rings. At a high level, the algorithms consist of some combination of the following three steps
\begin{enumerate}
\item Relaxing the promise CSP as a linear program over the rationals and solving it (possibly multiple times).
\item Relaxing the promise CSP as a system of linear equations over a commutative ring and solving it.
\item Combining these solutions using a rounding rule to obtain a solution to the promise CSP. 
\end{enumerate}

Although Step (1) may seem standard, solving the LP naively may encounter issues in Step (3). For example, if the rounding rule is of the form ``Round to 0 if less than $1/2$, and Round to 1 if greater than $1/2$,'' the algorithm will fail on the case the LP solver assigns $1/2$ to a variable. Previous works encountering these issues \cite{DBLP:journals/siamcomp/AustrinGH17, DBLP:conf/soda/BrakensiekG18} circumvent the problem with ad-hoc techniques (like adjusting the linear program solution slightly). We propose to work around this issue with a more principled technique: finding a solution to the linear program that lies in a completely different \emph{ring} like $\mathbb Z[\sqrt{2}]$. Since $1/2 \not\in \mathbb Z[\sqrt{2}]$ the issue disappears! We believe that this is the first application of such an idea to approximation algorithms.\footnote{A previous version of this paper falsely claimed that finding solutions to linear programs which lie in rings such as $\mathbb Z[\sqrt{2}]$ was previously done by Adler and Beling~\cite{DBLP:journals/algorithmica/AdlerB94,DBLP:journals/mp/AdlerB92}, but in those works the solutions lie in fields (like $\mathbb Q[\sqrt{2}]$).}

 Below we state a special case of this result when there is only one block, so that the polymorphisms are ``threshold-periodic'' \emph{symmetric} functions (for simplicity, this case is treated first in Section~\ref{sec:boolean}, before the more general block-symmetric case in Section~\ref{sec:Boolean-block}). 
 Namely, such polymorphisms look at the range of the Hamming weight of its input, based on which it applies a certain periodic function of the Hamming weight.  We stress that imposing a symmetry requirement on the polymorphisms is very different from imposing a symmetry condition on the predicates. At least for the Boolean domain, the latter was solved\footnote{This classification used an additional assumption that the predicates can be applied to negations of variables.} in our earlier work~\cite{DBLP:conf/soda/BrakensiekG18}, whereas the general dichotomy for symmetric polymorphisms is still open\footnote{In the case of CSPs, it is known that having a symmetric polymorphism of \emph{every} arity allows for LP relaxations \cite{DBLP:conf/innovations/KunOTYZ12, DBLP:conf/dagstuhl/BartoKW17}. Such results also generalize to promise CSPs, but they are only a special case of having infinitely many symmetric polymorphisms. In fact, almost none of the tractable examples in \cite{DBLP:conf/soda/BrakensiekG18} have symmetric polymorphisms of all arities.}, although this work is a substantial step in that direction.

\begin{thm}[Informal version of Theorem~\ref{thm:boolean-threshold-periodic}]\label{thm:inf}
  Let $E$ be a finite set, let $0 = \tau_0< \tau_1 < \dots < \tau_{k-1} < \tau_k=1$ be a sequence of rationals, let $M = (M_1, M_2, \hdots, M_k)$ a sequence of positive integers, and let $\eta_i : \mathbb Z/M_i\mathbb Z \to E$ be periodic functions for $i=1,\dots,k$.  Consider a promise CSP $\{(P_i, Q_i)\}$ with the $P_i$'s and $Q_i$'s defined over the domains $\{0, 1\}$ and $E$, respectively. Further suppose the promise CSP  admits a family of polymorphisms $f_L : \{0,1\}^L \to E$ for infinitely many $L$ such that
  \[ f_L(x) = \begin{cases}
      \eta_1(0) & \Ham(x) = 0\\
      \eta_i(\Ham(x)\!\!\!\!\!\!\mod M_i) & L\tau_{i-1} < \Ham(x) < L\tau_i, \ i=1,2,\dots,k \\
      \eta_{k}(L) & \Ham(x) = L.
    \end{cases}\] 
  where $\Ham(x)$ denotes the Hamming weight of $x$. Then the promise CSP can be solved in polynomial time.
\end{thm}

Although we still do not have a general dichotomy for symmetric polymorphisms, we boldly conjecture that \emph{every} finite Boolean Promise CSP with infinitely many symmetric polymorphisms has a threshold-periodic family of polymorphisms. This conjecture would imply that every finite Boolean Promise CSP with infinitely many symmetric polymorphisms is polynomial-time tractable. The main hurdle to proving this conjecture seems topological (or combinatorial) rather than algorithmic: showing that every infinite family of symmetric polymorphisms corresponding to a finite Boolean Promise CSP implies the existence of an infinite threshold-periodic family.

\smallskip \noindent \textbf{Linear Programming Result.} As discussed earlier, since we desire to have a rich family of rounding functions, we need to ensure that our linear program does not output a solution that lies on the rounding boundaries. A novel technique to avoid this is to find a solution to the linear program that lies in a different ring, such as $\mathbb Z[\sqrt{2}]$. This is proved in the following result.

\begin{thm}\label{thm:Z2}
  Let $Mx \le b$ be a linear program (without an objective function\footnote{Any result on solving a linear program without an objective can be extended to a result on approximately maximizing an objective by using a binary search protocol.}) over the rational numbers, with $M \in \mathbb Q^{m \times n}, b \in \mathbb Q^m$. Then, one can decide whether there is a solution $x \in \mathbb Z[\sqrt{q}]$, where $q$ is a non-square positive integer, in $\poly(n + m + \size)$ time, where $\size$ is the number of bits needed to represent the input.
\end{thm}

The theory required to prove this result is developed in Section~\ref{sec:LP}.

\smallskip \noindent \textbf{An Example.}
Returning to the algorithmic result, Theorem~\ref{thm:inf}, we examine a didactic example. Let $E = \{0, 1, 2, 3\}$ and consider a promise CSP with a single pair of predicates $(P, Q)$, which are defined to be
\begin{align*}
  P &= \{x \in \{0, 1\}^6 : \Ham(x) = 3\}\\
  Q &= \{y \in \{0, 1, 2, 3\}^6 : y_i \not\in \{0, 3\}^6 \cup \{1, 2\}^6\text{ and }\sum_{i=1}^6 y_i \equiv 1\mod 2.\}
\end{align*}
Note that $P \subseteq Q$, so $(P, Q)$ is a valid pair of predicates for a promise CSP.\footnote{In Section~\ref{sec:prelim}, we allow for a more general mapping $\phi : \{0, 1\} \to E$ such that $\phi(P) \subseteq Q$.} At first, it is unclear what algebraic structure $(P, Q)$ has, but it turns out for all odd $L$ to have the following polymorphism $g_L : \{0, 1\}^L \to \{0, 1, 2, 3\}$.
\begin{align*}
  g_L(x) &= \begin{cases}
    0 & \Ham(x) < L/2 \text{ and } \Ham(x) \equiv 0 \mod 2\\
    3 & \Ham(x) < L/2 \text{ and } \Ham(x) \equiv 1 \mod 2\\
    2 & \Ham(x) > L/2 \text{ and } \Ham(x) \equiv 0 \mod 2\\
    1 & \Ham(x) > L/2 \text{ and } \Ham(x) \equiv 1 \mod 2.
  \end{cases}
\end{align*}
In Theorem~\ref{thm:inf}, this corresponds to the choices $k = 2$, $\tau_1 = 1/2$, $M_1 = M_2 = 2$, $\eta_1(0) = 0$, $\eta_1(1) = 3$, $\eta_2(0) = 2$, and $\eta_2(1) = 1$. We leave as an exercise to the reader to check why this family of $g_L$'s are polymorphisms of $(P, Q)$. Below, we give an overview of our algorithm for this special case. This serves as an illustration of the crux of our strategy, which involves blending together two broad approaches underlying efficient CSP algorithms, namely linear programming and solving linear systems over rings.

This combining of linear programs and linear systems has important parallels in the CSP literature. At a high level, CSPs solvable by linear programming relaxations have a connection to ``bounded width'' constraint satisfaction problems (e.g., \cite{DBLP:conf/innovations/KunOTYZ12}) and CSPs representable as ring equations have Mal'tsev polymorphisms (e.g., \cite{DBLP:conf/dagstuhl/BartoKW17}). Thus, by ``synthesizing'' these two techniques, we are understanding promise CSPs (like the $(P, Q)$ just mentioned) which neither method by itself would resolve. Although the authors are currently unaware of a technical connection, this combining of bounded width and Mal'tsev techniques was the last hurdle that was needed to be overcome to resolve the CSP dichotomy~\cite{DBLP:conf/focs/Bulatov17,DBLP:conf/focs/Zhuk17}.

\subsection{Overview of ideas for a special case}

To give insight into the proof of Theorem~\ref{thm:boolean-threshold-periodic}, we give a high-level overview of how to solve promise CSPs using the predicate $(P, Q)$ mentioned in the previous subsection with $P \subset \{0, 1\}^6$ and $Q \subset \{0, 1, 2, 3\}^6$. As stated, there is an infinite family of threshold-periodic polymorphisms $g_L : \{0, 1\}^L \to \{0, 1, 2, 3\}$ (where $L$ is odd).

Imagine we have an instance of a CSP with constraints from $P$ on Boolean variables $x_1, \hdots, x_n$. We seek to find $y_1, \hdots, y_n \in \{0, 1, 2, 3\}^m$ which satisfies the corresponding CSP instance with respect to $Q$. We first construct a Basic LP relaxation.

\textbf{Basic LP Relaxation.} In the Basic LP relaxation, for each $x_i \in \{0, 1\}$ we consider a relaxed version $v_i \in [0, 1]$. For every constraint $P(x_{i_1}, \hdots, x_{i_6})$, we specify $(v_{i_1}, \hdots, v_{i_6})$ must live in the convex hull of $P$. We can find real-valued $v_i$'s which satisfy these conditions in polynomial time.

Now consider if we try to round the $v_i$'s right away. Consider a constraint $P(x_{i_1}, \hdots, x_{i_6})$, then we know there is a convex combination of elements of $P$ which equals $(v_{i_1}, \hdots, v_{i_6})$. A key idea introduced in our previous work \cite{DBLP:conf/soda/BrakensiekG18} was that the weights of the convex combination can, in the limit, be approximated by an average of the elements of $P$ using integer weights which sum to an odd number. Imagine this weighted average being arranged as a matrix
\begin{center}
  \begin{tabular}{cccccccc}
    & $x_{i_1}^{(1)}$ & $x_{i_2}^{(1)}$ & $x_{i_3}^{(1)}$ & $x_{i_4}^{(1)}$ & $x_{i_5}^{(1)}$ & $x_{i_6}^{(1)}$ & $\in P$\\
    & $x_{i_1}^{(2)}$ & $x_{i_2}^{(2)}$ & $x_{i_3}^{(2)}$ & $x_{i_4}^{(2)}$ & $x_{i_5}^{(2)}$ & $x_{i_6}^{(2)}$ & $\in P$\\
    & $\vdots$ & $\vdots$ & $\vdots$ & $\vdots$ & $\vdots$ & $\vdots$\\
    & $x_{i_1}^{(L)}$ & $x_{i_2}^{(L)}$ & $x_{i_3}^{(L)}$ & $x_{i_4}^{(L)}$ & $x_{i_5}^{(L)}$ & $x_{i_6}^{(L)}$ & $\in P$\\\hline
    Average & $\approx v_{i_1}$ & $\approx v_{i_2}$ & $\approx v_{i_3}$ & $\approx v_{i_4}$ & $\approx v_{i_5}$ & $\approx v_{i_6}$ &\\
    $g_L$ & $\hat{y}_{i_1}$ &$\hat{y}_{i_2}$ &$\hat{y}_{i_3}$ &$\hat{y}_{i_4}$ &$\hat{y}_{i_5}$ &$\hat{y}_{i_6}$ &
 \end{tabular}
\end{center}

The key observation is that since the $L$ rows have elements of $P$, we can apply the polymorphism $g_L$ to get an element of $(\hat{y}_{i_1}, \hdots, \hat{y}_{i_6}) \in Q_i$.

Now think about what happens to $x_{i_1}$. If $v_{i_1} > 1/2$ then if $L$ is sufficiently large and the integer weights sufficiently accurate, then the Hamming weight of the column $(x_{i_1}^{(1)}, \hdots, x_{i_1}^{(L)})$ will be greater than $L/2$, guaranteeing that $\hat{y}_{i_1}$ is $1$ or $2$. Likewise, if $v_{i_1} < 1/2$, then we can guarantee that $\hat{y}_{i_1}$ is either $0$ or $3$. We can deftly avoid the case $v_{i_1} = 1/2$ from ever happening, by solving the linear program over a subring of $\mathbb R$ that is dense but does not contain $1/2$, such as $\mathbb Z[\sqrt{2}]$.

Since the same variable can appear in many predicates in the instance, issues can arise. For the variable $x_{i_1}$ note that the Basic LP made a \emph{global} choice that either $v_{i_1} > 1/2$ or $v_{i_1} < 1/2$. Thus, for every clause that $x_{i_1}$ appears in, the corresponding $\hat{y}_{i_1}$ will always be in $\{0, 3\}$ (if $v_{i_1} < 1/2$) or $\{1, 2\}$ (if $v_{i_1} > 1/2$). However, this approach on its own cannot globally ensure that $\hat{y}_{i_1}$ is always equal to, say, $0$ instead of $3$. This due to the current lack of control on the \emph{parity} of how many times each element of $P$ shows up in the matrix above, since this parity is what the polymorphism $g_L$ looks at when deciding whether $\hat{y}_{i_1}$ is $0$ or $3$. Naive attempts to force a certain parity fail, as the same variable needs the same parity assigned across all the constraints it appears in. To repair this, we also need to consider the affine relaxation.

\textbf{Affine Relaxation.} Here, we let $V$ be the smallest affine subspace (with respect to $\mathbb F_2$) which contains $P$. Then, each constraint $P(x_{i_1}, \hdots, x_{i_m})$ is relaxed to $(r_{i_1}, \hdots, r_{i_m}) \in V$ where $r_i \in \mathbb F_2$. Solving these relaxed constraints can be done in polynomial time using Gaussian Elimination over $\mathbb F_2$.

The beauty of utilizing this second relaxation is that whenever we run into the dilemma of $\hat{y}_{i_1} \in \{0, 3\}$ or $\hat{y}_{i_1} \in \{1, 2\}$, we can break the uncertainty by always setting $y_{i_1}$ to be element with the same parity as $r_{i_1}$! The reason this works is subtle but powerful. When picking the integer weights of the elements of $P$, we also require that the Hamming weight of each column modulo $2$ is equal to the $r_{i_1}$'s. When $L$ is really large, changing the parity does not harm the approximation, so the ``binning'' of $\hat{y}_{i_1} \in \{0, 3\}$ or $\hat{y}_{i_1} \in \{1, 2\}$ still works via the Basic LP. But now the addition of these $r_{i_1}$'s via the affine relaxation further guarantees that across clauses, the $\hat{y}_{i_1}$ chosen always has consistent parity with $r_{i_1}$. Thus, the $\hat{y}_{i_1}$'s do indeed satisfy all the $Q$ constraints. This completes the proof that $(P, Q)$ is a tractable promise CSP template.

Note that each of these two relaxations was a ``lifting''  of the $(P, Q)$ problem into the Boolean-domain Gaussian elimination problem and the infinite-domain Basic LP relaxation. Section~\ref{sec:sandwich} more formally defines how this lifting process works (where we call it \emph{sandwiching}). 

\subsection{Organization}

In Section~\ref{sec:prelim}, we describe the notation used for CSPs and promise CSPs, particularly for polymorphisms. In Section~\ref{sec:sandwich}, we formally define the Basic LP and affine relaxations (and combined relaxations) of a promise CSP via a notion we call a \emph{homomorphic sandwich}. In Section~\ref{sec:LP}, describe how to find solutions in a rich family of rings, such as $\mathbb Z[\sqrt{2}]$, to linear programs with integer coefficients. In Section~\ref{sec:boolean}, we prove that having an infinite family of threshold-periodic polymorphisms implies tractability, proving ``warm up'' results for threshold polymorphisms and periodic polymorphisms along the way. In Section~\ref{sec:Boolean-block}, we show how to extend these reductions to block-symmetric functions known as \emph{regional} and \emph{regional-periodic} polymorphisms. In Section~\ref{sec:larger}, we extended these results to larger domains. In Section~\ref{sec:conclusion}, we describe the challenges in further developing the theory of promise CSPs. Appendix~\ref{app:hom-proofs} proves that the reductions to finite and infinite domains described in Section~\ref{sec:sandwich} are correct and efficient.

On a first reading, we recommend focusing on Section~\ref{sec:boolean} after skimming Sections~\ref{sec:prelim}, \ref{sec:sandwich} and \ref{sec:LP}.

\section{Preliminaries}\label{sec:prelim}

In this section, we include the important definitions and results in the constraint satisfaction literature. In order to accommodate both the theorist and the logician, we give the definitions from multiple perspectives.

\subsection{Constraint Satisfaction}\label{subsec:CSP}

In this paper, a constraint satisfaction problem consists of a \emph{domain} $D$ and a set $\Gamma = \{P_i \subseteq D^{\ar_i}: i \in I\}$ of \emph{constraints} or \emph{relations}. Each $\ar_i$ is called the \emph{arity} of constraint $P_i$ and the collection $\sigma = \{(i, \ar_i) : i \in I\}$ is called a \emph{signature}. We say that $(x_1, \hdots, x_{\ar_i})$ \emph{satisfies} a constraint $P_i$ if $(x_1, \hdots, x_{\ar_i}) \in P_i$. This is written as $P_i(x_1, \hdots, x_{\ar_i}).$  This indexed set of constraints $\Gamma$ is often referred to as the \emph{template}.\footnote{The tuple $(D, \sigma, \Gamma)$ of the domain, signature and template is known as a \emph{structure}.}

A \emph{$\Gamma$-CSP} is a formula written in conjunctive normal form (CNF) with constraints from $\Gamma$. That is, for some index set $J$
\[
  \Psi(x_1, \hdots, x_n) = \bigwedge_{j \in J} P_{i_j}(x_{j_1}, \hdots, x_{j_{\ar_{i_j}}}).
\]
We say that the formula is satisfiable if there is an assignment of variables which satisfies every clause. The decision problem $\CSP(\Gamma)$ corresponds to the language $\{\Phi : \text{$\Phi$ is a satisfiable $\Gamma$-CSP}\}.$ In other words, given $\Phi$, is it satisfiable?

\begin{rem}
  In the CSP literature, another common way to define a $\Gamma$-CSP is consider the domain $X = \{x_1, \hdots, x_n\}$ and a template $\Psi$ with signature $\sigma$. We say that $\Psi$ is satisfiable, if there is a \emph{homomorphism} (to be defined soon) $f : X \to D$ from $\Psi$ to $\Gamma$.
\end{rem}

The famous Dichotomy Conjecture of Feder and Vardi~\cite{DBLP:journals/siamcomp/FederV98} conjectured that for every finite domain $D$ and template $\Gamma$, $\Gamma$-CSP is either in $\mathsf P$ or in $\mathsf {NP}$-complete.

The case $|D| = 2$ was first fully solved by Schaefer~\cite{Schaefer:1978}. This was later extended to the case $|D| = 3$ by Bulatov~\cite{DBLP:journals/jacm/Bulatov06}, and finally general finite $D$ in the recent independent works by Bulatov~\cite{DBLP:conf/focs/Bulatov17} and Zhuk~\cite{DBLP:conf/focs/Zhuk17}. An extraordinarily important tool in the resolution of the Dichotomy conjecture is \emph{polymorphisms} (e.g., \cite{Chen2009, DBLP:conf/dagstuhl/BartoKW17}).

Given a relation $P \subseteq D^{\ar}$ and a function $f : D^L \to E$ (where $D$ and $E$ may be equal), we define $f(P)$ to be\footnote{This corresponds to the $O_f(P)$ notation from \cite{DBLP:conf/soda/BrakensiekG18}.}
\begin{align*}
  \{(f(x_1^{(1)}, \hdots, x_1^{(L)}), \hdots, f(x_{\ar}^{(1)}, \hdots, x_{\ar}^{(L)})) :
  x^{(1)}, \hdots, x^{(L)} \in P\}.
\end{align*}
More pictorially (c.f., \cite{DBLP:conf/dagstuhl/BartoKW17})
\begin{center}
  \begin{tabular}{lllll}
    $(x_1^{(1)},$ & $x_2^{(1)},$ & $\hdots,$ & $x_{\ar}^{(1)})$ & $\in P$\\
    $(x_1^{(2)},$ & $x_2^{(2)},$ & $\hdots,$ & $x_{\ar}^{(2)})$ & $\in P$\\
    $\vdots$ & $\vdots$ & & $\vdots$ &\\
    $(x_1^{(L)},$ & $x_2^{(L)},$ & $\hdots,$ & $x_{\ar}^{(L)})$ & $\in P$\\
    $\Downarrow f$ & $\Downarrow f$ & $\hdots$ & $\Downarrow f$ &\\
    $y_1$ & $y_2$ & $\hdots$ & $y_k$ & $\in f(P)$
  \end{tabular}
\end{center}

Given this notion, we can now define both what a \emph{homomorphism} and what a \emph{polymorphism} are.

\begin{df}
  Let $D$ and $E$ be domains and $\sigma = \{(i, \ar_i) : i \in I\}$ be a signature. Let $\Gamma = \{P_i \subseteq D^{\ar_i} : i \in I\}$ and $\Gamma' = \{P_i' \subseteq E^{\ar_i} : i \in I\}$ be templates with signature $\sigma$. A map $f : D \to E$ is a \emph{homomorphism} from $\Gamma$ to $\Gamma'$ if $f(P_i) \subseteq P_i'$ for all $i$.
\end{df}

As an example, consider $D = \{0, 1\}$ and $E = \{0, 1, 2\}$ and $\sigma = \{(1, 2)\}$. Consider $\Gamma_{2\text{--col}} = \{P_1 = \{(0, 1), (1, 0)\} \in D^2\}$ and $\Gamma_{3\text{--col}} = \{Q_1 = \{(0, 1), (0, 2), (1, 0), (1, 2), (2, 0), (2, 1)\} \in E^2\}$, which are the templates for $2$-coloring and $3$-coloring respectively. Then, the map $\id_D$ is a homomorphism from $\Gamma_{2\text{--col}}$ to $\Gamma_{3\text{--col}}$. In other words, any $2$-colorable graph is also $3$-colorable.

\begin{df}
  Let $D$ be a domain and  $\Gamma = \{P_i \subseteq D^{\ar_i}: i \in I\}$ be a template. A \emph{polymorphism} of a CSP is a function $f : D^L \to D$ for some positive integer $L$ such that $f(P_i) \subseteq P_i$ for all $i \in I$. We let $\Pol(\Gamma)$ denote the set of a polymorphisms of $\Gamma$.
\end{df}

Intuitively, polymorphisms are algebraic objects which combine solutions of CSPs to produce another solution. We now give a few standard examples.

\begin{enumerate}
\item Consider any template $\Gamma$. A trivial example of such an $f$ is a \emph{projection} function: for some $i \in [L] := \{1, \hdots, L\}$, for all $x_1, \hdots, x_L \in D$, we let $f(x_1, \hdots, x_L) = x_i$. This function is a polymorphism for every $\Gamma$. More generally, we say that a polymorphism is \emph{essentially unary} (or a \emph{dictator}) if $f$ depends on exactly one coordinate (in particular, this does not include constant functions).
\item Consider a polymorphism $f : D^L \to D$ such that $f \in \Pol(\Gamma)$. Let $\pi : [L] \to [R]$ be any surjective map, where $R \le L$ is a positive integer. Then, $f^{\pi} : D^R \to D$ is defined to be
  \[
    f^{\pi}(x_1, \hdots, x_R) = f(y_1, \hdots, y_L)\text{ where $y_j = x_{\pi(j)}$ for all $j \in L$.}
  \]
  We have that $f^{\pi} \in \Pol(\Gamma)$ (e.g., \cite{DBLP:conf/soda/BrakensiekG18}).
\item Linear Equations. Consider any finite field $\mathbb F$. Let \[\Gamma_{\mathbb F\text{--lin}} = \{P_i \subset \mathbb F^{\ar_i} : P_i\text{ affine subspace}\}\] be a template of linear constraints. Then, the map $f(x, y, z) = x - y + z$ is a polymorphism. In the case $\mathbb F = \mathbb F_2$, this is called $\PAR_3$.
\item 2-SAT. The template for 2-SAT can be expressed in a few ways, one is \[\Gamma_{2\text{--SAT}} = \{P_1 = \{(1, 1), (1, 0), (0, 1)\}, P_2 = \{(1, 0), (0, 1)\}\}.\] Then $\MAJ_{3}$, the majority function on $3$ bits, is a polymorphism (e.g., \cite{Chen2009}). 
\end{enumerate}

One reason polymorphisms are so fundamental, is due to an elegant property known as a \emph{Galois correspondence} (or \emph{Galois connection}). From a computational complexity perspective\footnote{From a logic perspective, there is a \emph{primitive-positive reduction} from $\Gamma_2$ to $\Gamma_1$.}, if two finite CSP templates $\Gamma_1$ and $\Gamma_2$ of the same domain, but not necessarily of the same signature, satisfy $\Pol(\Gamma_1) \subseteq \Pol(\Gamma_2)$, then there is a polynomial-time reduction from $\CSP(\Gamma_2)$ to $\CSP(\Gamma_1)$. Thus, from a computational complexity perspective, it is sufficient to think about the polymorphisms of a CSP rather than the individual constraints. We can now state Schaefer's theorem rather elegantly.

\begin{thm}[\cite{Schaefer:1978}, as stated in, e.g, \cite{Bulatov2005}]
  Let $D = \{0, 1\}$ and let $\Gamma$ be a template. $\CSP(\Gamma) \in \mathsf P$ if and only if\footnote{We assume in this paper that $\mathsf P \neq \mathsf {NP}$.} $\Pol(\Gamma)$ has a non-dictator polymorphism. Otherwise, $\CSP(\Gamma)$ is $\mathsf {NP}$-complete.
\end{thm}

\subsection{Promise Constraint Satisfaction}\label{subsec:PCSP}

Next, we discuss an approximation variant of CSPs known as \emph{promise CSPs} (or PCSPs), first studied systematically by the authors in \cite{DBLP:conf/soda/BrakensiekG18}. Intuitively, a promise CSP is just like a CSP except that the constraints have ``slack'' to them which allows for an algebraic form of approximation.

\begin{df}
  A \emph{promise domain} is a triple $(D, E, \phi)$, where $\phi$ is a map from $D$ to $E$.
\end{df}

The most commonly used promise domain in this article will be $D = E = \{0, 1\}$ and $\phi = \id_{D}$ is the identity map.

\begin{df} 
  Let $(D, E, \phi)$ be a promise domain and let $\sigma = \{(i, \ar_i) : i \in I\}$ be a signature. A \emph{promise template} $\Gamma = (\Gamma_P, \Gamma_Q)$is a pair of templates $\Gamma_P = \{P_i \in D^{\ar_i}\}$ and $\Gamma_Q = \{Q_i \in E^{\ar_i}\}$ each with signature $\sigma$ such that $\phi$ is a homomorphism from $\Gamma_P$ to $\Gamma_Q$. Each pair $(P_i, Q_i)$ is called a \emph{promise constraint}.
\end{df}

In the simplest case, $D = E$ and $\phi = \id_D$, then the homomorphism condition is equivalent to $P_i \subseteq Q_i$. Note that in general $\phi$ could be an injection, surjection or neither. 

\begin{df}
  Let $\Gamma = \{\Gamma_P, \Gamma_Q\}$ be a promise template over the promise domain $(D, E, \phi)$. A \emph{$\Gamma$-PCSP} is a pair of CNF formulae $\Psi_P$ and $\Psi_Q$ with identical structure. That is, there is an index set $J$ such that 
  \begin{align*}
    \Psi_P(x_1, \hdots, x_n) &= \bigwedge_{j \in J} P_{i_j}(x_{j_1}, \hdots, x_{j_{\ar_{i_j}}})\\
    \Psi_Q(y_1, \hdots, y_n) &= \bigwedge_{j \in J} Q_{i_j}(y_{j_1}, \hdots, y_{j_{\ar_{i_j}}})
  \end{align*}
\end{df}

\begin{rem}
  Just like a $\Gamma$-CSP, a $\Gamma$-PCSP can be expressed in the language of homomorphisms. Our domain again $X = \{x_1, \hdots, x_n\}$, and $\Psi$ is a template over the domain $X$ with the same signature as $\Gamma_P$ and $\Gamma_Q$. Satisfying $\Psi_P$ and $\Psi_Q$ corresponds to finding homomorphisms from $\Psi$ to $\Gamma_P$ and $\Gamma_Q$, respectively.
\end{rem}

Note that if $x_1, \hdots, x_n$ satisfies $\Psi_P$, then $\phi(x_1), \hdots, \phi(x_n)$ satisfies $\Psi_Q$. In particular, satisfying $\Psi_Q$ is ``easier'' (in a logical, not algorithmic, sense) than satisfying $\Psi_P$. Thus, we can define a promise problem.

\begin{df}[Promise CSP--decision version]
  Let $\Gamma$ be a promise CSP. We define $\PCSP(\Gamma)$ to be the following promise decision problem on promise formulae $(\Psi_P, \Psi_Q)$. 
  \begin{itemize}
  \item \textbf{ACCEPT:} $\Psi_P$ is satisfiable.
  \item \textbf{REJECT:} $\Psi_Q$ is not satisfiable.
  \end{itemize}
\end{df}

This has a corresponding search variant

\begin{df}[Promise CSP--search version]
  Let $\Gamma$ be a promise CSP. We define Search-$\PCSP(\Gamma)$ to be the following promise search problem on promise formulae $(\Psi_P, \Psi_Q)$.
  \begin{itemize}
  \item Given that $\Psi_P$ is satisfiable, output a satisfying assignment to $\Psi_Q$.
  \end{itemize}
\end{df}

Unlike classical CSPs, in which the decision and search versions are often\footnote{This requires that fixing a variable to a specific value leads to a constraint with the same polymorphisms.} polynomial-time equivalent, it is not clear that the decision and search variants of $\PCSP(\Gamma)$ have the same computational complexity, although there no known $\Gamma$ for which the complexity differs. Even so, there is a reduction from the decision version to the search version: run the algorithm for the search version, and check if it satisfies $\Psi_Q$.

The following are interesting examples of promise CSPs, (c.f., \cite{DBLP:conf/soda/BrakensiekG18}). In the first five examples, the domain is Boolean: $(D = \{0, 1\}, D, \id_D)$.

\begin{enumerate}
\item \textbf{CSPs.} Let $\Gamma = \{P_i \in D^{k_i}\}$ be a CSP over the domain $D$, then $(D, D, \id_D)$ is a promise domain and $\Lambda = (\Gamma, \Gamma)$ is a promise CSP.
\item \textbf{(2+$\epsilon$)-SAT.} Let $\NEQ = \{(0, 1), (1, 0)\}$. Fix, a positive integer $k$, and let $P_1 = \{x \in D^{2k+1} : \Ham(x) \ge k\}$ and $Q_1 = \{x \in D^{2k+1} : \Ham(x) \ge 1\}$. Then, $\Gamma = (\{P_1, \NEQ\}, \{Q_1, \NEQ\})$ corresponds to a promise variant of $(2k+1)$-SAT: if every clause in a $(2k+1)$-SAT instance is true for at least $k$ variables, can one find a ``normal'' satisfying assignment. This problem was shown to be $\mathsf{NP}$-hard by Austrin, Guruswami, and H{\aa}stad~\cite{DBLP:journals/siamcomp/AustrinGH17}. 
\item \textbf{Threshold conditions.} Fix $\alpha, \beta$ such that $1/\alpha + 1/\beta = 1$. Let $\Gamma_{P} = \{P_i \subset D^{\ar_i} : i \in \{1, 2\}\}$ and $\Gamma_{Q} = \{Q_i \subseteq D^{\ar_i} : \{1, 2\}\}$, where the promise constraints are as follows (the choices of $\ar_1, \ar_2, s, t$ are arbitrary).
  \begin{align*}
    P_1 &= \{x \in D^{\ar_1} : \Ham(x) \le s\} & Q_1 &= \{x \in D^{\ar_1} : \Ham(x) \le \alpha s\}\\
    P_2 &= \{x \in D^{\ar_2} : \Ham(x) \ge \ar_2 - t\} & Q_2 &= \{x \in D^{\ar_2} : \Ham(x) \ge \ar_2 - \beta t\}.
  \end{align*}
  We have that $\PCSP(\Gamma_P, \Gamma_Q)$ is tractable (this is also alluded to in \cite{DBLP:conf/soda/BrakensiekG18}).
\item \textbf{Sandwiching Linear Equations.} Let $A \subseteq \mathbb F_7^{\ar}$ be an affine subspace. Specify a map $h : \mathbb F_7 \to \{0, 1\}$ such that $h(0) = 0$ and $h(1) = 1$. Let $\Gamma_P = \{A \cap D^{\ar}\}$ and $\Gamma_Q = \{h(A)\}$. Then, $(\Gamma_P, \Gamma_Q)$ is tractable by performing Gaussian elimination over $\mathbb F_7$, even though the domain is Boolean! This type of promise CSP was briefly alluded to in~\cite{DBLP:conf/soda/BrakensiekG18}, and in this work we systematically study such examples though the theory of \emph{homomorphic sandwiches} (see Section~\ref{subsec:sandwich}).
\item \textbf{Hitting Set.} Let $\Gamma_P := \{P_i \in D^{\ar_i} : i \in I\},$ where $P_i := \{x \in D^{\ar_i} : \Ham(x) = \ell_i\}$ where $\ell_i \in \{1, \hdots, \ar_i-1\}$. In other words, $\CSP(\Gamma_P)$ corresponds to a generalized hitting set problem: given a collection of hyperedges $S_i$ and targets $\ell_i$, find a subset of the vertices $S$ such that $|S \cap S_i| = \ell_i$. Let $\Gamma_Q := \{Q_i = D^{\ar_i} \setminus \{0^{\ar_i}, 1^{\ar_i}\}\}$. Then $\CSP(\Gamma_Q)$ is hypergraph two-coloring (each color appears at least once per hyperedge). Although neither $\CSP(\Gamma_P)$ or $\CSP(\Gamma_Q)$ is tractable in general, $\PCSP(\Gamma_P, \Gamma_Q)$ is tractable~\cite{DBLP:conf/soda/BrakensiekG18}.
\item \textbf{Approximate Graph Coloring.} Let $k \le \ell$ be positive integer. Let $D = [k]$ and $E = [\ell]$. Then, $(D, E, \id_D)$ is a promise domain. Let $\Gamma_{k\text{--col}} = \{P = \{(x, y) \in D^2 : x \neq y\}\}$ and $\Gamma_{\ell\text{--col}} = \{Q = \{(x, y) \in E^2 : x \neq y\}\}$. Then, $\Gamma = (\Gamma_{k\text{--col}}, \Gamma_{\ell\text{--col}})$ is then the promise template for the well-studied \emph{approximate graph coloring problem}: given a graph of chromatic number $k$, find an $\ell$-coloring. This problem has been studied for decades, and it is still unsolved in many cases (e.g., \cite{DBLP:journals/siamdm/GuruswamiK04,DBLP:journals/combinatorica/KhannaLS00,DBLP:conf/approx/Huang13,DBLP:conf/coco/BrakensiekG16}).
\item \textbf{Rainbow Coloring.} Consider $D = [k]$ and $E = \{0, 1\}$, with $\phi : D \to E$ being an arbitrary, non-constant map. If $\Gamma_P = \{P = \{x \in D^k : x\text{ is a permutation of }D\}\}$ and $\Gamma_Q = \{Q = E^k \setminus \{0^k, 1^k\}\}$, then $\PCSP(\Gamma_P, \Gamma_Q)$ is the following hypergraph problem: given a $k$-uniform hypergraph such that there is a $k$-coloring in which every color appears in every edge, find a hypergraph $2$-coloring. A random-walk-based polynomial-time algorithm is reported in \cite{McDiarmid1993}; semidefinite programming gives another folklore algorithm. A deterministic algorithm based on solving linear programming was found by Alon [personal communication].
\end{enumerate}

Just as polymorphisms are useful for studying CSPs, they are a powerful tool for understanding promise CSPs. The first formal definition of a polymorphism of a promise CSP appeared in \cite{DBLP:journals/siamcomp/AustrinGH17}.

\begin{df}
  Let $(D, E, \phi)$ be a promise domain and $\sigma = \{(i, \ar_i) : i \in I\}$ be a signature. and $\Gamma = (\Gamma_P = \{P_i \subseteq D^{\ar_i}\}, \Gamma_Q = \{Q_i \subseteq E^{\ar_i}\})$ be a promise template over this domain. A \emph{polymorphism} of a promise CSP is a function $f : D^L \to E$ for some positive integer $L$ such that for all $i \in I$, $f(P_i) \subseteq Q_i$. We let $\Pol(\Gamma)$ denote the set of polymorphisms of $\Gamma$.
\end{df}

Like in the case of constraint satisfaction problems, there is a Galois correspondence for promise CSPs (\cite{DBLP:conf/soda/BrakensiekG18}, following from a result of \cite{Pippenger2002}). Thus, like for CSPs, it suffices to consider the collection of polymorphisms and not the particulars of the constraints.

As pointed out in \cite{DBLP:journals/siamcomp/AustrinGH17}, unlike the polymorphisms for ``ordinary'' CSPs, due to the change in domain of polymorphisms for promise CSPs, they cannot be composed. Thus, unlike CSPs which deal with families of CSPs closed under compositions and projections\footnote{Here a projection, also known as a \emph{minor}, is an identification of some subsets of the coordinates. } (known as \emph{clones}), promise CSPs are determined by families of CSPs closed under only projections\footnote{There is an additional constraint known as \emph{finitization}, which comes as a technicality when $\Gamma$ is finite. See the Conclusion for a discussion on how this could be utilized to understand promise CSPs from a topological perspective.}. Thus, while the techniques for studying CSPs are (universal) \emph{algebraic}, the necessary techniques for studying promise CSPs are \emph{topological}. In particular, unlike results such as Schaefer's theorem for which the existence of one nontrivial polymorphism (e.g., $\PAR_3$) is enough to imply tractability of a CSP, \cite{DBLP:journals/siamcomp/AustrinGH17} showed that an infinite sequence of polymorphisms of a promise CSP is necessary to imply tractability. One contribution of this work is that we show that having an infinite sequences of polymorphisms which ``converge'' with respect to a particular topology is sufficient to imply efficient algorithms.  We now give a polymorphic reason for why each of the above examples is tractable/non-tractable.

\begin{enumerate}
\item \textbf{CSPs.} The polymorphisms of the CSP template $\Gamma$ are exactly the same as the polymorphisms of the promise template $(\Gamma, \Gamma)$. 

\item \textbf{(2+$\epsilon$)-SAT.} \cite{DBLP:journals/siamcomp/AustrinGH17} showed that $\MAJ_{2k-1}$ is a polymorphism, but $\MAJ_{2k+1}$ (or any function that essentially depends on at least $2k+1$ variables) is not. They exploited this fact to show $\mathsf{NP}$-hardness via a reduction from the PCP theorem.
\item \textbf{Threshold conditions.} Consider $L$ such that $L/\alpha$ is not an integer. Then,
  \[
    f_L(x_1, \hdots, x_L) = \begin{cases}
      0 & \Ham(x) < L/\alpha\\
      1 & \Ham(x) > L/\alpha
    \end{cases}
  \]
  is a polymorphism of this problem. This is known as a \emph{threshold polymorphism} and is studied in Section~\ref{subsec:threshold}.
\item \textbf{Sandwiching Linear Equations.} Consider $L \equiv 1 \mod 7$ then
  \[
    f_L(x_1, \hdots, x_L) = h\left(\sum_{i=1}^L x_i \mod 7\right)
  \]
  is a family of polymorphisms. This is a \emph{periodic polymorphism}, and it is studied in Section~\ref{subsec:periodic}.
\item \textbf{Hitting Set.} \cite{DBLP:conf/soda/BrakensiekG18} showed that for all odd integers $L$
  \[
    \AT_L(x_1, \hdots, x_L) = \mathbf 1[x_1 - x_2 + x_3 - \cdots - x_{L-1} + x_L \ge 1],
  \]
  is a polymorphism for this problem. In this paper, we have generalized polymorphisms like these to \emph{regional polymorphisms}, which are studied in Section~\ref{sec:Boolean-block}.

\item \textbf{Approximate Graph Coloring.} As this question is still open, much is not yet understood about the polymorphisms. \cite{DBLP:conf/coco/BrakensiekG16} showed that when $\ell \le 2k - 2$, then the polymorphisms ``look'' dictatorial when restricted to some subset of the outputs. These polymorphisms are closely connected to the independent sets of \emph{tensor powers of cliques} (e.g., \cite{Alon}).
\item \textbf{Rainbow Coloring.} This problem has many, many nontrivial polymorphisms. For example, for odd $L$, $f_L : [k]^L \to \{0, 1\}$ defined to be
  \[
    f_L(x_1, \hdots, x_L) = \mathbf 1\left[\sum_{i=1}^L x_i \le \frac{2k+3}{4}\cdot L\right],
  \]
  is a family of polymorphisms for this problem. This is an example of a non-Boolean \emph{regional polymorphism}, which is studied in Section~\ref{sec:larger}.
\end{enumerate}

The main result of our previous work~\cite{DBLP:conf/soda/BrakensiekG18} is as follows.
  
\begin{thm}[\cite{DBLP:conf/soda/BrakensiekG18}]
\label{thm:BG18-main}
  Consider the promise domain $(D = \{0, 1\}, D, \id_D)$. Let $\Gamma = (\Gamma_P = \{P_i \in D^{\ar_i}\}, \Gamma_Q = \{Q_i \in D^{\ar_i}\})$ be a CSP with the following technical conditions.
  \begin{itemize}
  \item $P_1 = Q_1 = \{(0, 1), (1, 0)\}$. In other words, variables can be negated.
  \item For all $i \in I$, $P_i$ and $Q_i$ are \emph{symmetric}: if $(x_1, \hdots, x_{ar_i}) \in R_i$ then $(x_{\pi(1)}, \hdots, x_{\pi(ar_i)}) \in R_i$ for all permutations $\pi$. 
  \end{itemize}
  Then, either $\PCSP(\Gamma)$ is (promise) $\mathsf{NP}$-hard or $\PCSP(\Gamma)$ is in $\mathsf{P}$ and has one of the following $6$ infinite sequences of polymorphisms (coming in $3$ pairs).
  \begin{enumerate}
  \item $\MAJ_L$ for all odd $L \ge 3$ or\footnote{The negation symbol $(\neg)$ in front a polymorphism merely means to negate the output.} $\neg \MAJ_L$ for all odd $L \ge 3$.
  \item $\PAR_L$ for all odd $L \ge 3$ or $\neg \PAR_L$ for all odd $L \ge 3$.
  \item $\AT_L$ for all odd $L \ge 3$ or $\neg \AT_L$ for all odd $L \ge 3$.
  \end{enumerate}
\end{thm}

Although that paper tried to jointly understand both algorithmic and hardness results (with most of the work coming in on the hardness side), the aim of this paper is to develop the algorithmic tools for understanding tractable cases of promise CSPs. In particular, we more deeply explore the power of LP and affine (i.e., linear equations over a commutative ring) relaxations for solving promise CSPs.

\section{Relaxing Promise CSPs with Homomorphic Sandwiches}\label{sec:sandwich}

In this section, we build on the theory described in Section~\ref{sec:prelim} to rigorously connect promise CSPs with \emph{relaxations} of these problems (e.g., linear programming relaxations).

\subsection{The Homomorphic Sandwich}\label{subsec:sandwich}

Often it is useful to reduce promise CSP to a tractable (promise) CSP in another domain, and then map the result back to the original domain. We call this procedure a \emph{homomorphic sandwich}.

\begin{df}
  Let $(D, E, \phi)$ be a promise domain and $\sigma = \{(i, \ar_i) : i \in I\}$ be a signature. Let $\Gamma = (\Gamma_P = \{P_i \subseteq E^{\ar_i}\}, \Gamma_Q = \{Q_i \subseteq D^{\ar_i}\})$ be a promise template with this signature. Let $F$ be a (possibly infinite) domain. Let $\Lambda$ be a template over $F$ with signature $\sigma$. We say that $\Gamma$ \emph{homomorphically sandwiches} $\Lambda$ if there exist maps $g : D \to F$, $h : F \to E$ such that $g$ is a homomorphism from $\Gamma_P$ to $\Lambda$ and $h$ is a homomorphism from $\Lambda$ to $\Gamma_Q$.
\end{df}

\begin{rem}
  Note that we can also say that a promise template $(\Gamma_P, \Gamma_Q)$ homomorphically sandwiches another promise template $(\Lambda_P, \Lambda_Q)$ if there are homomorphisms from $\Gamma_P$ to $\Lambda_P$ and $\Lambda_Q$ to $\Gamma_Q$. In the remainder of the paper, we assume that $\Lambda$ is a CSP template unless otherwise specified.
\end{rem}

\begin{rem}
For brevity, we often say that $\Gamma$ sandwiches $\Lambda$ if $\Gamma$ homomorphically sandwiches $\Lambda$. Furthermore, we often do not want to a-priori restrict $\Lambda$ to have a particular signature, so we more generally say that $\Gamma$ sandwiches a set of relations $\Lambda$ if there is a subset $\Lambda'$ with the same signature as $\Gamma$ so that $\Gamma$ sandwiches $\Lambda'$. 
\end{rem}

To exemplify this definition, we give a few examples of homomorphic sandwiches.

\begin{enumerate}
\item Let $\Gamma$ be any CSP over $D$, then the promise CSP $(\Gamma, \Gamma)$ sandwiches $\Gamma$ via $(g, h) = (\id_D, \id_D)$
\item Recall from Example 6 of Section~\ref{subsec:PCSP} that $\Gamma = (\Gamma_{k\text{--col}}, \Gamma_{\ell\text{--col}})$ is the promise template for the $k$ vs. $\ell$ approximate graph coloring problem. For any $m \in \{k, \hdots, \ell\}$, we have that $\Gamma$ sandwiches $\Gamma_{m\text{--col}}$ via $(\id_{[k]}, \id_{[m]})$. In other words, any algorithm which solves the $m$-coloring problem can also solve the $k$ vs. $\ell$ approximate graph coloring problem.
\item Consider $\Gamma = (\Gamma_P, \Gamma_Q)$ from Example 4 of Section~\ref{subsec:PCSP} with affine subspace $A \le \mathbb F^{\ar}_7$ and the map $h : \mathbb F_7 \to \{0, 1\}$. Then, $\Gamma$ sandwiches $\Lambda$ via $(\id_{\{0, 1\}}, h)$.
\end{enumerate}

In practice, $D$ and $E$ will both be finite, but $F$ is likely infinite. Thus, $g : D \to F$ is a finite map (often something canonical, like the identity function) which tells how to express our promise CSP in the new domain. On the other hand, $h : F \to E$, which we call the \emph{rounding function}, is where the ``algorithmic magic'' takes place.\footnote{Or, keeping with the sandwich theme, $h$ is the \emph{special sauce}.} When $F$ is infinite, it is not \emph{a priori} obvious that $h$ has a computationally efficient description, so we often assume that we have \emph{oracle access} to this rounding function. Furthermore, the choice of rounding function $h$ is crucially tied to the polymorphisms of $\Gamma$. In fact, one can consider $h$ to be the ``limit'' of a sequence of polymorphisms of $\Gamma$, or alternatively $\Pol(\Gamma)$ is a \emph{discretization} of $h$. This connection between $h$ and polymorphisms is made more clear in Section~\ref{sec:boolean}.

To be the best of the authors' knowledge, all known tractable promise CSPs have the property that the given promise template $\Gamma$ sandwiches a judiciously chosen $\Lambda$ for which $\CSP(\Lambda)$ is polynomial-time tractable. In fact, the authors conjecture that any tractable promise CSP must sandwich some (possibly infinite) tractable CSP.

However, even though $\Gamma$ has a finite domain, $\Lambda$ often necessarily has infinite domain, even when $\Gamma$ is Boolean. Very recently, Barto [unpublished] showed that $\Lambda$ must be infinite in the tractable case that $\Gamma = (\text{1-in-3-SAT}, \text{NAE-3SAT})$ studied in \cite{DBLP:conf/soda/BrakensiekG18}. The Basic LP and affine relaxations, explained in the following sections, are instances of sandwiching an infinite-domain $\Lambda$. This is another reason why classifying the tractability of promise CSPs is so much more difficult than for ordinary CSPs, and perhaps partially explains the difficulty of theory community's struggle to resolve the approximate graph coloring problem.

\subsection{Basic LP Relaxation}\label{subsec:basiclp}

The Basic LP relaxation is a widespread tool in approximation algorithms, often giving optimal results. For example, the resolution of the Finite-Valued CSP (VCSP) dichotomy due to Thapper and {\v Z}ivn{\'y}~\cite{DBLP:journals/jacm/ThapperZ16} showed that all tractable instances can be solved with a Basic LP relaxation.  In \cite{DBLP:conf/soda/BrakensiekG18}, the Basic LP was one of the classes of algorithms exhibited in tractable promise CSPs (used for the $\MAJ$ and $\AT$ families). In this work, we vastly generalize the usage of such an algorithm.

Fix a template $\Gamma = \{P_i \subseteq D^{\ar_i}\}$ and consider an instance of $\CSP(\Gamma)$

\[
  \Psi(x_1, \hdots, x_n) := \bigwedge_{j \in J} P_{i_j}(x_{j_1}, \hdots, x_{j_{\ar_{i_j}}}).
\]

Fix a subring\footnote{In this paper, all subrings will have the element $1$. In particular, every subring of $\mathbb R$ contains $\mathbb Z$.} $A \subset \mathbb R$ which is to be the domain of our Basic LP. (Typically, $A = \mathbb Q$, but for reasons we are soon to see, other commutative rings are useful.)

Fix a positive integer $k \ge 1$ and a map $g : D \to A^k$. Then, for each $P_i \in \Gamma$, $g(P_i)$ is a cloud of points in $(A^k)^{\ar_i} \subseteq \mathbb R^{k\ar_i}$. Recall the notion of a \emph{convex hull} of a set of points $S \in \mathbb R^n$
\[
  \Conv(S) = \left\{\sum_{i=1}^\ell \alpha_iz_i : \alpha_i \in [0, 1], z_i \in S, \sum_{i=1}^{\ell}\alpha_i = 1\right\}.
\]
We let $\Conv_A(S) = \Conv(S) \cap A^n$. If we assume that each $P_i$ has constant size, then $\Conv(g(P_i))$ can be specified by a constant number of linear inequalities.

The following is the Basic LP relaxation.

\begin{framed}
    \begin{itemize}
    \item Input: $\Psi(x_1, \hdots, x_n) := \bigwedge_{j \in J} P_{i_j}(x_{j_1}, \hdots, x_{j_{\ar_{i_j}}})$, an instance of $\CSP(\Gamma)$.
    \item Variables: each $x_i$ is replaced by $v_i\in A^k$.
    \item Constraints:
      \begin{itemize}
      \item For each $x_i$, specify that $v_i \in \Conv_{A^k}(g(D))$.
      \item For each constraint $P_{i_j}(x_{j_1}, \hdots, x_{j_{\ar_{i_j}}})$ specify that
        \[
          (v_{j_1}, \hdots, v_{j_{\ar_{i_j}}}) \in \Conv_{A^k}(g(P_{i_j})).
        \]
      \end{itemize}
    \end{itemize}
   \centering{\textbf{Relaxation \thesection.1.} The Basic LP relaxation of $\CSP(\Gamma)$.}
\end{framed}
\begin{rem}
Since our primary goal is feasibility, our LP relaxations do not have objective functions.
\end{rem}

\subsubsection{LP-solvable rings}

If we assume that $\Gamma$ is finite, each $P_i$ has constant size, the size of this LP relaxation is linear in the size of the input $\Psi$, with constant factors depending on the specific $\Gamma$.  Note that if $A = \mathbb Q$, we can test feasibility and output a solution in polynomial time. Like most uses of linear programming in approximation algorithms, an LP solution, once found, is \emph{rounded} to solve the problem at hand. Due to technical restrictions of the rounding algorithms in this paper, there are often edge cases, for which rounding will not work.  For example, the procedure ``round to the nearest integer'' does not work for $v_{i,j} = 1/2$. In \cite{DBLP:conf/soda/BrakensiekG18}, the authors used an ad-hoc approach for avoiding these $1/2$ situations, but it turns out these can be solved in a more principled manner by solving the LP over a different ring other than $\mathbb Q$! Of course, linear programs over certain rings such as $A = \mathbb Z$, are not solvable in polynomial time unless $\mathsf{P} = \mathsf{NP}$, so we need to look at so-called \emph{LP-solvable} rings.

\begin{df}\label{df:LP-solv}
  Let $A \subset \mathbb R$ be an efficiently computable subring (defined in Section~\ref{sec:LP}). Let $Mx \le b$ be a system of linear inequalities with $M \in \mathbb Z^{m \times n}, b \in \mathbb Z^n$. If a feasible point $x \in A^n$ can be computed in $\poly(n + m + \size)$ time, where $\size$ is the representation complexity of the system of inequalities, then $A$ is \emph{LP-solvable}.
\end{df}

Note that we only consider linear programs with integer coefficients (equivalently rational coefficients), as all Basic LP reductions considered in this paper will have that form.

Thus, $A = \mathbb Q$ is LP-solvable, but $A = \mathbb Z$ is not. For technical reasons, we assume that $\mathbb R$ is \emph{not} LP-solvable because, in general, elements of $\mathbb R$ do not have a computable description. Even with these restrictions, there is still a plethora of $A$ which are LP-solvable. Results of Adler and Beling~\cite{DBLP:journals/mp/AdlerB92,DBLP:journals/algorithmica/Beling01}, show that fields such as bounded-degree algebraic extensions of $\mathbb Q$, such as $\mathbb Q[\sqrt{q}]$ for $q$ non-square are LP-solvable, although this may only be for the unit-cost arithmetic model.\footnote{An earlier version of this manuscript incorrectly interpreted one of their results to mean that $\mathbb Z[\sqrt{q}]$ is LP-solvable, but this conflated the coefficients of the linear equations ($\mathbb Z[\sqrt{q}]$ or equivalently $\mathbb Q[\sqrt{q}]$) with the domain the feasible points reside in ($\mathbb Q[\sqrt{q}]$).}

In Section~\ref{sec:LP}, we prove Theorem~\ref{thm:Z2} follows from the fact that $\mathbb Z[\sqrt{q}]$ is LP-solvable for non-square $q > 0$. The usefulness of this fact is that edge cases like rounding $1/2$ can be avoided by solving the LP over, say, the ring $\mathbb Z[\sqrt{2}]$. The authors are unaware of a previous application of this fact to approximation algorithms of CSPs.

\subsubsection{Sandwich interpretation}

Returning to main task of solving Promise CSPs, this procedure of rewriting a promise CSP as a Basic LP and then rounding is really a homomorphic sandwich. Let $A \subset \mathbb R$ be an LP-solvable ring and let $\Gamma$ be a promise relation over the promise domain $(D, E, \phi)$. As first introduced in Section~\ref{subsec:sandwich}, let $g : D \to \mathbb Z^k$ be any map, and let $h : A^k \to E$ be our rounding function. Let $\LP_{A^k} = \{R : R = \Conv_{A^k}(S), S \subset \mathbb Z^{k\ell}\text{ finite}, \ell \ge 1\}$ be the family of convex subsets.

\begin{restatable}{thm}{basiclp}\label{thm:basiclp}
  Let $A \subset \mathbb R^k$ be an LP-solvable ring. Let $(D, E, \phi)$ be a finite promise domain. Let $\Gamma = (\Gamma_P = \{P_i \in D^{\ar_i} : i \in I\}, \Gamma_Q = \{Q_i \in E^{\ar_i}\})$ be a finite promise CSP. Assume that $\Gamma$ sandwiches $\LP_{A^k}$ via $(g : D \to \mathbb Z^k, h : A^k \to E)$. Then $\PCSP(\Gamma) \in \mathsf{P}^h$, in which $\mathsf{P}^h$ is the family promise languages which can be computed in polynomial time given oracle access to $h$.
\end{restatable}

This result is proven in Appendix~\ref{app:hom-proofs}. Intuitively, Theorem~\ref{thm:basiclp} abstracts away the fine details of working with the Basic LP and reduces the task to showing the existence of a sandwich.

\subsection{Affine Relaxation}\label{subsec:affine}

Another broad class of algorithms studied in the CSP literature correspond to solving system of linear equations over some commutative ring. Linear-equation-solving algorithms are captured in the CSP literature under the broader class of CSPs with a \emph{Mal'tsev polymorphism}: a function on three variables such that $\varphi(x, y, y) = \varphi(y, y, x) = x$ always. Such CSPs are known to be tractable (e.g., \cite{DBLP:conf/dagstuhl/BartoKW17}). For commutative rings, the canonical Mal'tsev polymorphism is $\varphi(x, y, z) = x - y + z$, when the domain is a finite ring. Such algorithms are not restricted to finite domains: linear equations over $\mathbb Q$ can be solved in polynomial time using Gaussian elimination, and linear equations over $\mathbb Z$ can be solved in polynomial time by computing the Hermite Normal Form (e.g., \cite{frumkin1976algorithm, DBLP:journals/siamcomp/KannanB79, grotschel2012geometric}). This leads to the natural notion of \emph{LE-solvable} rings.

\begin{df}
  Define a commutative ring $R$ to be \emph{LE-solvable}, if systems of linear equations over $R$ can be efficiently solved in (weakly) polynomial time.
\end{df}

By the discussion above, all finite commutative rings are LE-solvable, as well as the infinite rings $\mathbb Z^k$ for any natural number $k$. Furthermore, every LP-solvable ring is LE-solvable as LPs are more expressive than LEs. Just as LPs relax sets to their convex hulls, linear equations relax sets to their \emph{affine hulls}. Given a subset $S \subseteq R$ and a subring $R' \subset R$, define the \emph{affine hull} to be
\[
  \Aff_{R'}(S) = \left\{r_1s^1 + \cdots + r_ks^k : \forall j, s^j \in S, r_j \in R'\text{ and }\sum_{j=1}^k r_j = 1\right\}.
\]
The reason we want the $r_j$'s to be restricted to a subring is made apparent in Section~\ref{sec:larger}, where we need to ensure that the $r_j$'s are integers (i.e., multiples of $1 \in R$). Note that by design $S \subseteq \Aff_{R'}(S)$. Furthermore, if $S$ is finite, checking whether $x \in \Aff_{R'}(S)$ can constrained by two linear conditions:
\begin{align*}
  x &= r_1s^1 + \cdots + r_ks^k\\
  1 &= r_1 + \cdots + r_k,
\end{align*}
This system of equations is a bit strange in that we are specifying that $x \in R$ but $r_1, \hdots, r_k \in R'$. As long as both $R'$ and $R$ are finite, this system still has the Mal'tsev polymorphism as the operator $(x, y, z) \mapsto x - y + z$ is closed under every ring. More generally, we say that the pair $(R', R)$ is LE-solvable.

We can now define the \emph{affine relaxation} of any CSP. Fix a finite domain $D$ and a ring $R$ and a subring $R'$. Also pick a map $g : D \to R$ to be our homomorphism.\footnote{Note that $g$ is not necessarily a \emph{ring} homomorphism.} Let $\Gamma = \{P_i \subseteq D^{\ar_i}\}$ be any template over $D$. Note that $g(P_i)$ is some finite subset of $R^{\ar_i}$. This leads to our description of an affine relaxation.

\begin{framed}
\begin{itemize}
\item Input: $\Psi(x_1, \hdots, x_n) := \bigwedge_{j \in J} P_{i_j}(x_{j_1}, \hdots, x_{j_{\ar_{i_j}}})$, instance of $\CSP(\Gamma)$, 
\item Variables: each $x_i$ is replaced by $w_{i} \in R$.
\item Constraints: For each constraint $P_{i_j}(x_{j_1}, \hdots, x_{j_{\ar_{i_j}}})$ of $\Psi$, specify that
    \[
      (w_{j_1}, \hdots, w_{j_{\ar_{i_j}}}) \in \Aff_{R'}(g(P_{i_j})).
    \]
  \end{itemize}
  \centering{\textbf{Relaxation \thesection.2.} The affine relaxation of $\CSP(\Gamma)$ with respect to a pair of rings $R' \subset R$.}
\end{framed}

By definition, $\Aff_{R'}(g(P_i))$ has a constant-sized description, since each $P_i$ is of constant size, and there are finitely many possible values for $g(P_i)$, a lookup table of the linear constraints can be formed. Thus, the system can be generated in linear time, and so it can be solved in polynomial time whenever $(R', R)$ is LE-solvable. Let

\[\LE_{R' \subset R} = \{\Aff_{R'}(S) : \exists k, S \subseteq R^k, S\text{ finite}\}.\]

If $R' = R$, then we denote this as $\LE_{R}$. This leads to an analogue of Theorem~\ref{thm:basiclp} for the infinite template over $R$

\begin{restatable}{thm}{affine}\label{thm:affine}
  Let $(R', R)$ be an LE-solvable pair. Let $(D, E, \phi)$ be a finite promise domain, and let $\Gamma = (\Gamma_P = \{P_i \in D^{\ar_i}\}, \Gamma_Q = \{Q_i \in E^{\ar_i}\})$ be a finite promise CSP over this promise domain. Assume that $\Gamma$ sandwiches $\LE_{R' \subset R}$ via $(g : D \to R, h : R \to E)$. Then, $\PCSP(\Gamma) \in \mathsf{P}^h$.
\end{restatable}

This result is proven in Appendix~\ref{app:hom-proofs}.

\subsection{Combined Relaxation}\label{subsec:combined}

The true power of these homomorphic sandwiches is revealed when these relaxations are combined using \emph{direct products}.

Consider CSP templates $\Lambda_1$ and $\Lambda_2$ over domains $F_1$ and $F_2$ (not necessarily finite), respectively. We define the \emph{direct product} $\Lambda_1 \times \Lambda_2$ to be the CSP template over the domain $F_1 \times F_2$ such that
\[
  \Lambda_1 \times \Lambda_2 = \{R_1 \times R_2 \subseteq (F_1\times F_2)^{\ar}: R_1\in \Lambda_1, R_2 \in \Lambda_2\text{ same arity }\ar\},
\]
where $(R_1 \times R_2)((x_1, y_1), \hdots, (x_{\ar}, y_{\ar})) = R_1(x_1, \hdots, x_{\ar}) \wedge R_2(y_1, \hdots, y_{\ar}).$ 

Note that up to relabeling coordinates, the direct product is commutative and associative, allowing the seamless combination of two or more CSP templates.

Fix a sequence of LP-solvable rings $\mathcal A := (A_1, \hdots, A_\ell)$ and a sequence of LE-solvable pairs $\mathcal R := (R'_1 \subset R_1, \hdots, R'_m \subset R_m)$. Now define the template
\[
  \LPLE_{\mathcal A, \mathcal R} := \LP_{A_1^{k_1}} \times \cdots \times \LP_{A_{\ell}^{k_{\ell}}}  \times \LE_{R'_1 \subset R_1} \times \cdots \times \LE_{R'_m \subset R_m}.
\]
It turns out promise homomorphisms to this template correspond to algorithms which \emph{combine} linear programming and affine equation solving.

\begin{restatable}{thm}{basiclpplusaffine}\label{thm:basiclp+affine}
  Let $\mathcal A := (A_1^{k_1}, \hdots, A_{\ell}^{k_{\ell}})$ be a sequence of LP-solvable rings, and let $\mathcal R := (R'_1 \subset R_1, \hdots, R'_m \subset R_m)$ be a sequence of LE-solvable pairs. Let $(D, E, \phi)$ be a finite promise domain, and let $\Gamma = (\Gamma_P, \Gamma_Q)$ be a finite promise CSP over this domain. Assume that $\Gamma$ sandwiches $\LPLE_{\mathcal A, \mathcal R}$ via $(g : D \to \mathbb Z^{k_1} \times \cdots \times \mathbb Z^{k_{\ell}} \times R_1 \times \cdots \times R_m,\,\,\, h : A_1^{k_1} \times \cdots \times A_{\ell}^{k_{\ell}} \times R_1 \times \cdots \times R_m \to E)$. Then, $\PCSP(\Gamma) \in \mathsf{P}^h$.
\end{restatable}

This result is proven in Appendix~\ref{app:hom-proofs}. Theorems \ref{thm:basiclp}, \ref{thm:affine}, and \ref{thm:basiclp+affine} are useful in that if we can show for a particular promise template $\Gamma$ sandwiches a suitable $\LP_{A}, \LE_{R'\subset R}$ or $\LPLE_{\mathcal A, \mathcal R}$ via $(g, h)$ \emph{and} $h$ is proven to be polynomial-time computable, then we can show that $\PCSP(\Gamma) \in \mathsf{P}$.

Before we establish some homomorphic sandwiches from the perspective of polymorphisms in Sections~\ref{sec:boolean}-\ref{sec:larger}, we take a closer look at LP-solvability.

\section{Understanding LP-solvability}\label{sec:LP}

In this section, we try to understand LP-solvability, particularly for the rings $\mathbb Z[\sqrt{q}]$. First, we need to discuss the subrings of $\mathbb R$ to which our results apply.

\subsection{Efficiently computable rings}

Let $A \subset \mathbb R$ be the subring which we are examining. We assume that each element $a \in A$ has some underlying representation as bits. We let $\size_A(a)$ denote the number of bits needed to represent $a$ with respect to the ring $A$. It may be the case that a element has multiple representations, (e.g., non-simplified fractions), in which case $\size_A$ refers to the bit representation given as input.

For example, if $A = \mathbb Z$, then $\size_{\mathbb Z}(n) = O(\log (|n|+1))$. Or, if $A = \mathbb Q$, a common choice is to have $\size_{\mathbb Q}(a/b) = O(\log (|a| + |b| + 1)).$ For a more detailed discussed about representation complexity (in the context of the rationals), see Chapter~1 of \cite{grotschel2012geometric}.

\begin{df}
  A subring $A \subset \mathbb R$ is an \emph{efficiently computable subring} if the following conditions hold.

  \begin{enumerate}
  \item For all $n \in \mathbb Z$, there is an efficiently computable representation of $n$ (with respect to $A$) such that $\size_A(n) \le \poly(\log (|n| + 1))$.
  \item For all $a \in A$ with $a > 0$,
    \[
      2^{\poly(\size_A(a))} > a > 2^{-\poly(\size_A(a))}.
    \]
  \item Each of the binary operations $\{+, -, \times,  =, <\}$ is computable in polynomial time (as a function of the input bit complexities).
  \item There is a universal constant $C > 0$ such that For each $\circ \in \{+, -, \times\}$ and for all $a, b \in \mathbb R$, $\size_A(a \circ b) \le C(\size_A(a) + \size_A(b))$.
  \end{enumerate}
  
  Likewise, we say that a subfield $F \subset \mathbb R$ is an \emph{efficiently computable subfield}, if the operator $\div$ can be added to conditions (3) and (4).
\end{df}

We quickly note that any efficiently computable subring can be extended to an efficiently computable subfield.

\begin{prop}\label{prop:quot}
  Let $A \subset \mathbb R$ be an efficiently computable subring. Then, $\quot(A) := \{a/b \mid a,b \in A, b \neq 0\}$ is an efficiently computable subfield.
\end{prop}

\begin{proof}
  Represent elements of $\quot(A)$ as pairs of elements of $A$ (where the second element of the pair is always nonzero). The operators can then be implemented using the following identities.
  \begin{align*}
    \frac{a}{b} + \frac{c}{d} &= \frac{a\times d + b\times c}{b\times d}\\
    \frac{a}{b} - \frac{c}{d} &= \frac{a\times d - b\times c}{b\times d}\\
    \frac{a}{b} \times \frac{c}{d} &= \frac{a\times c}{b\times d}\\
    \frac{a}{b} \div \frac{c}{d} &= \frac{a\times d}{b\times c} \text{ (if $c \neq 0$)}\\
    \frac{a}{b} = \frac{c}{d} &\iff a \times d = b \times c\\
    \frac{a}{b} < \frac{c}{d} &\iff a \times d < b \times c \text{ (if $b, d > 0$)}.
  \end{align*}
  By inspection, these meet the conditions of an efficiently computable subfield.
\end{proof}

From this we can now state the main result of this section.

\begin{thm}\label{thm:ecds}
  Let $\mathbb Z \subsetneq A \subset \mathbb R$ be an efficiently computable subring. If $A$ is also LE-solvable, then $A$ is LP-solvable.
\end{thm}

We show in Section~\ref{subsec:sqrt} that $\mathbb Z[\sqrt{q}]$ is LE-solvable and is an efficiently computable dense subring, whenever $q$ is nonsquare, implying that it is LP-solvable.

\subsection{Proof of Theorem~\ref{thm:ecds}}

In this subsection, we prove Theorem~\ref{thm:ecds}. First, we start off by proving/stating a few folklore ingredients.

\subsubsection{Affine hulls and interior points}

Let $Mx \le b$ be the original linear program with coefficients in $\mathbb Q$. Let $K = \{x \in \mathbb R^n : Mx \le b\}$. One important step in our algorithm will be to compute a system $M'x = b'$ whose solution space is $\Aff_{\mathbb R}(K)$. Assuming that $K \neq \emptyset$, then $K$ has a nonzero interior with respect to the topology induced by $\Aff_{\mathbb R}(K)$. This is known as the \emph{relative interior}. The following result shows that both computing the affine hull and finding a point in the relative interior can be done in polynomial time on a Turing machine.

\begin{lem}[Lemma~6.5.3+Remark~6.5.4~\cite{grotschel2012geometric}]\label{lem:aff+int}
  Given a linear system $Mx \le b$ with $M \in \mathbb Q^{m\times n}, b \in \mathbb Q^{m}$ whose set of solutions is $K$, one can compute in $\poly(n + m + \size_\mathbb Q(M) + \size_\mathbb Q(b))$-time a system $M'y = b'$ whose solution space is $\Aff(K)$ as well as a point $y_0$ in the relative interior of $K$.
\end{lem}

Although it is easy to come up with ad-hoc approaches to prove this theorem (using linear programming over $\mathbb Q$ as a black-box) it is rather nontrivial to ensure that the bit complexity of the intermediate calculations is polynomial.

\subsubsection{More on efficiently computable subrings}

Since we started with an atomic description of what is efficiently computable, we need a lemma which shows that more complex arithmetic formulas are efficiently computable.

\begin{lem}\label{lem:dot-prod}
  Let $A$ be an efficiently computable ring. Let $a_1, \hdots, a_n, b_1, \hdots, b_n \in A$. Then, there exists $C' \ge 0$ such that
  \begin{enumerate}
  \item $\size_A(a_1 + \cdots + a_n) \le O(n^{C'}\sum_i \size_A(a_i))$.
  \item $\size_A(a_1 \times \cdots \times a_n) \le O(n^{C'}\sum_i \size_A(a_i))$.
  \item $\size_A(a_1 \times b_1 + \cdots + a_n\times  b_n) \le O(n^{C'}\sum_i [\size_A(a_i) + \size_A(b_i)])$.
  \end{enumerate}
  Thus, these expressions can be computed in polynomial time.
\end{lem}

\begin{proof}
  We give a proof of the first bound, the remaining cases are similar. By padding with $a_i = 0$, we may assume that $n = 2^k$ for some integer $k$ ($n$ will increase by at most a factor of two, so all bounds still hold, possibly with a worse constant factor).

  Recall that for all $a, b \in A$, $\size_A(a + b) \le C(\size_A(a) + \size_A(b))$. Thus, we inductively argue that
  \begin{align*}
    \size_A(a_1 + \cdots + a_{2^k}) &\le C\size_A(a_1 + \cdots + a_{2^{k-1}}) + C\size_A(a_{2^{k-1}+1} + \cdots + a_{2^k})\\
                                    &\le C^k\sum_{i=1}^{2^k}\size_A(a_i).
  \end{align*}
  Note that $C^k = (2^k)^{\log_2 C}$. Thus, setting $C' = \log_2 C$, we get the desired bound.
\end{proof}

Another important ingredient is showing that for any interval $(p, q)$ where $p$ and $q$ are rational, we can efficiently find $a \in A \cap (p, q)$. The following lemma shows that this is possible.

\begin{lem}\label{lem:dense}
  Let $A$ be an efficiently computable subring which is not $\mathbb Z$. Then, for all $p, q \in \quot(A)$ with $p < q$, one can efficiently compute $a \in A \cap (p, q)$ in $O(\size_\mathbb Q(p) + \size_{\mathbb Q}(q))$-time. 
\end{lem}

\begin{proof}
  If $0 \in (p, q)$ we can output $0$. If $p, q < 0$, we can find $a \in (-q, -p)$ and then output $-a$. Thus, we may assume without loss of generality $p, q > 0$.
  
  Since $p < 2^{\poly(\size_A(p))}$ (property 2 of an efficiently computable subring), we can compute $\lfloor p \rfloor$ in polynomial time via binary search (property 1 guarantees that integers can be efficiently computed). Thus, by subtracting out by this floor (and adding it back at the end), we may assume that $p \in [0, 1)$. If $q > 1$, then we can output $a = 1$. Thus, we may assume that $q \in (0, 1]$.

  Now, fix a constant $\alpha_0 \in A\cap (1/2, 2/3).$ (This must exist since every subring of $\mathbb R$ except for $\mathbb Z$ is dense in $\mathbb R$.) Now run the following binary-search like procedure:
  \begin{enumerate}
  \item[--] Input: $p, q \in \quot(A)$, $0 \le p < q \le 1$.
  \item[--] Output: $a \in A \cap (p, q)$.
  \item Initialize $a \ot 0$.
  \item For $i = 1, \hdots$
    \begin{enumerate}
    \item Compute $\alpha^i_0$.
    \item If $a + \alpha^i_0 < q$, set $a \ot a + \alpha^i_0$.
    \item If $p < a$, \textbf{Output} $a$.
    \end{enumerate}
  \end{enumerate}
  
  We claim that this for loop stops by $\log_{\alpha_0}(q-p) + 1$ steps. It suffices to prove by induction that after step $i$, $a \in (q - \alpha^i_0, q)$. After step 0 (i.e., before step 1), clearly $a = 0 \in (q - 1, q)$. If after step $i$, we know by the induction hypothesis that $a \in (q - \alpha^i_0, q)$, we then have two cases to consider for what happens in step $i+1$. First, if $a + \alpha^{i+1}_0 < q$, then \[a + \alpha^{i+1}_0 \in (q - \alpha^i_0 + \alpha^{i+1}_0, q) = (q - (1-\alpha_0)\alpha^i_0, q) \subset (q - \alpha^{i+1}_0, q)\]
  since $\alpha_0 \ge 1 - \alpha_0$. Otherwise, if $a + \alpha^{i+1}_0 > q$, this implies that $a > q - \alpha^{i+1}_0$, so $a \in (q - \alpha^{i+1}_0, q)$

  Thus, $a > p$ in at most $\log_{\alpha_0}(q - p) + 1$ steps. By property 2 of an efficiently computable subring, $\log_{\alpha_0}(q-p)$ is bounded by a polynomial in $O(\size_{\mathbb Q}(p) + \size_{\mathbb Q}(q))$, and thus the algorithm runs in polynomial time.
\end{proof}

\begin{rem}
  Note that we are iteratively adding to a variable, so it may be the case that the representation size of $a$ grows exponentially quickly. Lemma~\ref{lem:dot-prod} shows that the additions can be reordered so that the representation complexity of $a$ stays polynomially bounded. Thus, if the $+$ operation is not associative (in the sense of representations), one needs to add an extra step to recompute $a$ from scratch (keeping note of which $i$'s were used) whenever $\size(a)$ gets too large.
\end{rem}

\subsubsection{The algorithm and analysis}

Now, we state the general algorithm for finding solutions to LPs in efficiently computable subrings $\mathbb Z \subsetneq A \subset \mathbb R$.

\begin{framed}
   \begin{enumerate}
   \item[--] Input: $n, m, M \in \mathbb Q^{m \times n}, b \in \mathbb Q^m$.
   \item[--] Output: $x \in A^n$ such that $x \in K := \{y \in \mathbb R^n : My \le b\}$ or no solution exists.
   \item Compute $(M', b')$ such that $\Aff(K)$ is the solution set of $M'x = b'$.
   \item Compute $y_0 \in \mathbb Q^n$ in the relative interior of $K$. If no such $y_0$ exists, \textbf{Reject}.
   \item Compute whether $M'x = b'$ has a solution $x_0 \in A^n$. If no solution exists, \textbf{Reject}.
   \item Find an orthogonal basis $\{q_i \in \mathbb Z^n\}$ of the vector space $\Aff_{\mathbb R}(K) - y_0$. 
   \item Compute $\alpha_i \in \quot(A)$ such that $y_0 - x_0 =\sum_i \alpha_i q_i$. Note that $\alpha_i = \frac{(y_0 - x_0)\cdot q_i}{\|q_i\|^2}$ as $q_i$'s orthogonal.
   \item Compute $\epsilon := 2^{-\poly(n + m + \size_{\mathbb Q}(M) + \size_{\mathbb Q}(b))}$ and for all $i$, find $\beta_i \in (\alpha_i - \epsilon, \alpha_i + \epsilon) \cap A$.
   \item Set $z_0 = x_0 + \sum_i \beta_iq_i$. \textbf{Accept} and output $z_0$.
   \end{enumerate}
  
   \centering{\textbf{Algorithm \thesection.1.} Algorithm which demonstrates that $A$ is LP-solvable.}
\end{framed}

\begin{proof}[Proof of Correctness.]
  First, we establish that every step can be done efficiently. The efficiency of steps 1 and 2 follow from Lemma~\ref{lem:aff+int}. The efficiency of step 3 follows from the fact that $A$ is LE-solvable. Step 4 is efficient by the Gram-Schmidt algorithm (e.g., \cite{grotschel2012geometric}) and then multiplying each obtained basis vector by the least common denominator. Step 5 is efficient since $\quot(A)$ is efficiently computable (Proposition~\ref{prop:quot}) and dot products can be computed in polynomial time (Lemma~\ref{lem:dot-prod}). Step 6 can be computed efficiently by Lemma~\ref{lem:dense}. Step 7 can be computed efficiently by Lemma~\ref{lem:dot-prod}.

  Now that we know the algorithm is efficient, we show that it outputs the correct answer. The accuracy of Steps 1 and 2 is guaranteed by Lemma~\ref{lem:aff+int}. If there is no rational solution to the LP, then there cannot be a real solution, so the rejection is Step 2 is valid. In Step 3, any solution to the linear program in $A^n$ must also belong to the (real) affine hull of $K$, so the rejection in Step 3 is valid. For steps 4-6, note that since $x_0 \in \Aff_{\mathbb R}(K)$ and $\{q_i\}$ are integral vectors in the vector space $\Aff_{\mathbb R}(K) - y_0 = \Aff_{\mathbb R}(K) - x_0$, and combination of the form $\sum_i \beta_iq_i + x_0$ must be in $\Aff_{\mathbb R}(K) \cap A^n$. Since $y_0$ is in the relative interior and the boundary of $K$ is described by linear equations of bounded complexity, there must exist $\delta > 0$ of polynomial complexity such that $B_{L^2}(y_0, \delta) \cap \Aff(K) \subset K$ (see, e.g., the methods in Chapter 6 of \cite{grotschel2012geometric}). Then, if we set $\epsilon = \frac{\delta}{n\max_i\|q_i\|^2}$, any $\beta_i \in (\alpha_i - \epsilon, \alpha_i + \epsilon)$ will satisfy $z_0 := x_0 + \sum_i \beta_iq_i \in B_{L^2}(y_0, \delta) \cap \Aff_{\mathbb R}(K) \subset K$. Thus, since we can select these $\beta_i$ to be in our ring $A$, we have that the $z_0 \in A^n \cap K$, as desired.
\end{proof}

\subsection{Application to \texorpdfstring{$\mathbb Z[\sqrt{q}]$}{Z[sqrt(q)]}}\label{subsec:sqrt}

In this section, we show that $\mathbb Z[\sqrt{q}]$ is LP-solvable, whenever $q$ is a non-square.

We represent as $a + b\sqrt{q}$ as the pair of integers $(a, b)$ written in binary. It is not hard to then show that this representation makes $\mathbb Z[\sqrt{q}]$ an efficiently computable subring. The only tricky operation is $<$ which can be done efficiently by noting that
\[
  a + b\sqrt{q} > c + d\sqrt{q} \iff (a - c) > (d - b)\sqrt{q},
\]
which can then be checked by squaring and looking at the signs of $a - c$ and $d - b$.

The only bit left to prove is that $\mathbb Z[\sqrt{q}]$ is LE-solvable, which follows from the following lemma.

\begin{lem}
  Let $q$ be a non-square positive integer. Then, $\mathbb Z[\sqrt{q}]$ is LE-solvable.
\end{lem}

This result is almost certainly in the literature as $\mathbb Z[\sqrt{q}]$ is a finitely generated $\mathbb Z$-module (e.g., \cite{DBLP:books/lib/Cohen93}), but we provide a more elementary argument for completeness.

\begin{proof}
  Consider a system of equations $Mx = b$. with $M \in (\mathbb Z[\sqrt{q}])^{m \times n}$ and $b \in (\mathbb Z[\sqrt{q}])^m$. Then, we can express $M = M_1 + M_2\sqrt{q}$ where $M_1, M_2 \in \mathbb Z^{m \times n}$ and $b = b_1 + b_2\sqrt{q}$ where $b_1, b_2 \in \mathbb Z^m$.

  Thus, $Mx = b$ has a solution if and only if there exist $y, z \in \mathbb Z^n$ such that
  \[
    (M_1 + M_2\sqrt{q})(y + z\sqrt{q}) = b_1 + b_2\sqrt{q}.
  \]
  Since $\sqrt{q}$ and $1$ are linearly independent over the rationals, the above system has a solution if and only if the below system of equations with coefficients in $\mathbb Z$ has a solution
  \begin{align*}
    M_1y + qM_2z &= b_1\\
    M_2y + M_1z &= b_2.
  \end{align*}
   As $\mathbb Z$ is LE-solvable (\cite{DBLP:journals/siamcomp/KannanB79}), we must have that $\mathbb Z[\sqrt{q}]$ is LE-solvable.
\end{proof}

Thus, $\mathbb Z[\sqrt{q}]$ is LP-solvable for all nonsquare $q$, establishing Theorem~\ref{thm:Z2}. This is sufficient to obtain the results in the subsequent section.

\section{Threshold-Periodic Polymorphisms}\label{sec:boolean}

In this section and the subsequent one, we assume that our promise domain $(D, E, \phi)$ satisfies $D = \{0, 1\}$ and $E$ is any finite domain with any inclusion map $\phi : D \to E$. Restricting $D$ to be Boolean allows for a simplified presentation, the results of Sections~\ref{sec:boolean} and \ref{sec:Boolean-block} can be extended to larger domains, as described in Section~\ref{sec:larger}.

\subsection{Threshold Polymorphisms}\label{subsec:threshold}

Many polymorphisms which are considered in classical CSP theory, such as the $\OR$, $\AND$, and $\MAJ$ functions, can be thought of as \emph{threshold functions}. That is, the value of each of these polymorphisms only depends on whether the Hamming weight of the input is above a certain threshold. In this subsection, we consider a generalization of such functions to multiple thresholds.

\begin{df}
  A \emph{threshold sequence} is a finite sequence of rational\footnote{These could also be real numbers under suitable computational assumptions, but for simplicity of exposition we assume all thresholds are rational} numbers $\tau_0 = 0 < \tau_1 < \cdots < \tau_k = 1$.
\end{df}

For $x \in \{0, 1\}^L$, we let $\Ham(x)$ be the Hamming weight of $x$, i.e., the number of bits of $x$ set to $1$.

\begin{df}
  Let $T = \{\tau_0, \tau_1, \hdots, \tau_k\}$ be a threshold sequence and $\eta : \{0, 1, \hdots, k+1\} \to E$ be any map. Let $L$ be a positive integer such that $L\tau_i$ is not an integer for any $i \in \{1, \hdots, k - 1\}$. Then, define $\THR_{T,\eta,L} : \{0, 1\}^L \to \{0, 1\}$ to be the following polymorphism.
  \[\THR_{T,\eta,L}(x) = \begin{cases}
      \eta(0) & \Ham(x) = 0\\
      \eta(i) & L\tau_{i-1} < \Ham(x) < L\tau_i, 1 \le i \le k\\
      \eta(k+1) & \Ham(x) = L.
  \end{cases}\]
\end{df}

The function $\eta$ is closely connected to the rounding function $h$ from the definition of a homomorphic sandwich (Section~\ref{subsec:sandwich}). In essence, $\eta$ is finite description or discretization of $h$.

To get intuition, here are examples of common polymorphisms and their corresponding parameters as threshold functions. 
\begin{center}
\begin{tabular}{c|ccc}
  & $\MAJ_L$ & $\OR_L$ & $\AND_L$\\\hline
  $T$ & $\{0, 1/2, 1\}$ & $\{0, 1\}$ & $\{0, 1\}$\\
  $\eta$ & $(0, 0, 1, 1)$ & $(0, 1, 1)$ & $(0, 0, 1)$
\end{tabular}
\end{center}

This now leads to our first main result.

\begin{thm}\label{thm:boolean-threshold}
  Let $T = \{\tau_0, \hdots, \tau_k\}$ be a threshold sequence with a corresponding map $\eta : \{0, \hdots, k+1\} \to E$. Let $\Gamma = (\Gamma_P = \{P_i \in D^{\ar_i} \in I\}, \Gamma_Q=\{Q_i \in E^{\ar_i} : i \in I\})$ be a promise template such that $\THR_{T,\eta,L} \in \Pol(\Gamma)$ for infinitely many $L$. Then, $\PCSP(\Gamma) \in \mathsf{P}$.
\end{thm}

The proof is essentially a direct generalization of the arguments in Section~3.2 of \cite{DBLP:conf/soda/BrakensiekG18}.

\begin{proof}
   Let $A = \mathbb Z[\sqrt{2}]$, which as previously stated is LP-solvable by a theorem of Adler and Beling~\cite{DBLP:journals/algorithmica/AdlerB94}. We claim that  $\Gamma$ sandwiches $\LP_{A}$ via the following maps\footnote{For type-theoretic reasons, we define $h$ on the full domain of $A$ rather than  $[0, 1] \cap A$.} $g : D \to A$ and $h : A \to E$:
  \begin{align*}
    g(d) &= d\\
    h(v) &= \begin{cases}
      \eta(0) & v \le 0\\
      \eta(i) & \tau_{i-1} < v < \tau_i, 1 \le i \le k\\
      \eta(k+1) & v \ge 1
    \end{cases}
  \end{align*}

  Define $\LP_{\Gamma} := \{R_i := \Conv_A(g(P_i)) : P_i \in \Gamma_P\}$. Since $g(P_i) \subset \Conv_A(g(P_i))$, we have that $g$ is a homomorphism from $\Gamma_P$ to $\LP_{\Gamma}$. We claim that $h$ is a homomorphism from $\LP_{\Gamma}$ to $\Gamma_Q$. In other words, we seek to show that $h(\Conv_A(g(P_i))) \subseteq Q_i$. For any $V \in \Conv_A(g(P_i))$, since $g(P_i)$ is finite, there exist elements $X^1, \hdots,X^m \in P_i$ and weights\footnote{Note that the weights might not be in $A$.} $\alpha_1, \hdots, \alpha_m \in (0, 1]$ summing to $1$ such that \[V = \alpha_1g(X^1) + \cdots + \alpha_mg(X^m).\]

  Fix $L$ to be sufficiently large (to be specified later). We can pick nonnegative integers $w_1, \hdots, w_m$ such that $w_1 + \cdots + w_m = L$ and $|w_i - \alpha_i L| \le 1$ (start with $w_i = \lfloor \alpha_i L\rfloor$ for all $i$ and then increase weights one-by-one until the sum is $L$). Now compute
  \[
    Y := \THR_{T, \eta, L}(\underbrace{X^1, \hdots, X^1}_{w_1\text{ copies}}, \hdots, \underbrace{X^m, \hdots, X^m}_{w_m\text{ copies}}) \in Q_i.
  \]
  Define for each coordinate $j \in \{1, \hdots, \ar_i\}$
  \[
    s_j := \frac{1}{L}\Ham(\underbrace{X^1_j, \hdots, X^1_j}_{w_1\text{ copies}}, \hdots, \underbrace{X^m_j, \hdots, X^m_j}_{w_m\text{ copies}}).
  \]
  Then, by design, $Y_j = h(s_j)$. We also know that
  \[
    \frac{1}{L}|s_j - V_j|= \sum_{\substack{a = 1\\X^a_j = 1}}^m \left|\frac{w_a}{L} - \alpha_a\right| \le \frac{m}{L}
  \]
  Since $V \in A^{\ar_i}$, we have three cases
  \begin{enumerate}
  \item If $V_j \le 0$, then $X^a_j = 0$ for all $j$. Then, $V_j = s_j = 0$ so $Y_j = h(V_j)$.
  \item If $V_j \ge 1$, then $X^a_j = 1$ for all $j$. Then, $V_j = s_j = 1$ so $Y_j = h(V_j)$.
  \item If $V_j \in (0, 1)$, then $\tau_{i-1} < V_j < \tau_i$ for some $i \in \{1, \hdots, k\}$. Thus, as $L \to \infty$, $s_j$ will get sufficiently close to $V_j$ that $\tau_{i-1} < s_j < \tau_i$. Thus,  $Y_j = h(s_j) = h(V_j)$.
  \end{enumerate}
  Thus, since $Y \in Q_i$, we have that $h(V) \in Q_i$, establishing the homomorphic sandwich is valid.

  By Theorem~\ref{thm:basiclp}, we have that $\PCSP(\Gamma) \in \mathsf{P}^h$. Note that $h$ is polynomial-time computable, since computing $h$ involves checking a constant number of inequalities in $A$. Thus, $\PCSP(\Gamma) \in \mathsf{P}$, as desired.
\end{proof}

\subsection{Periodic Polymorphisms}\label{subsec:periodic}

Instead of having our threshold functions be piece-wise constant, we can consider periodic polymorphisms.

\begin{df}
  Let $M$ be a positive integer, and let $\eta : \mathbb Z/M\mathbb Z \to E$ be any map. Let $L$ be a positive integer. Define $\PER_{M, \eta, L}$ to be the following function
  \[
    \PER_{M,\eta,L}(x) = \eta(k)\text{ if }\Ham(x) \equiv k \mod M.
  \]
\end{df}

As stated earlier, Example 4 from Section~\ref{subsec:PCSP} is a periodic polymorphism.

\begin{thm}\label{thm:boolean-periodic}
  Let $M$ be a positive integer, and let $\eta : \mathbb Z/M\mathbb Z \to E$ be any function. Let $\Gamma = (\Gamma_P = \{P_i \subseteq D^{\ar_i}\}, \Gamma_Q = \{Q_i \subseteq E^{\ar_i}\})$ be a promise template on the Boolean domain such that $\PER_{M,\eta,L} \in \Pol(\Gamma)$ for infinitely many $L$. Then, $\PCSP(\Gamma) \in \mathsf{P}$.
\end{thm}

\begin{proof}
  For these infinitely many $L$, consider the remainders when they are divided by $M$. Since there are only finitely many remainders, there exists $r \in \{0, \hdots, M-1\}$ such that $L \equiv r \mod M$ infinitely often. 
  
  Consider the ring $R = \mathbb Z/M\mathbb Z$. We seek to show that $\Gamma$ sandwiches $\LE_R$ via the maps $g : \{0, 1\} \to R$ and $h : R \to E$ where
  \begin{align*}
    g(x) &= \begin{cases}
      0 & x = 0\\
      r & x = 1
    \end{cases}\\
    h &= \eta.
  \end{align*}

  Note that $\eta$ is a ``discretization'' of $h$, but since $R$ is a finite domain, $\eta$ can be used for $h$.

  Consider $\LE_{\Gamma} = \{R_i := \Aff(g(P_i)) : P_i \in \Gamma_P\}$. Since $g(P_i) \subseteq \Aff(g(P_i))$, we have that $g$ is a homomorphism from $\Gamma_P$ to $\LE_{\Gamma}$.  We claim that $h$ is a homomorphism from $\LE_{\Gamma}$ to $\Gamma_Q$. In other words, for all $(P_i, Q_i) \in \Gamma$, we seek to show that $h(\Aff(g(P_i))) \subseteq Q_i$. For any $V \in \Aff(g(P_i))$, we have that there exist $X^1, \hdots, X^k \in P_i$ as well as ring elements $r_1, \hdots, r_k \in R$ such that $r_1 + \cdots + r_k = 1$ and
  \[
    V = r_1g(X^1) + \hdots + r_kg(X^k).
  \]
  For some sufficiently large $L \equiv r \mod M$ for which $\PER_{M, \eta, L} \in \Pol(\Gamma)$, pick nonnegative integers $w_1, \hdots, w_k$ such that $w_i \equiv r_ir \mod M$ and $w_1 + \cdots + w_k = L$. By starting with the $w_i$'s as small as possible and then increment by $M$, this is possible as long as $L \ge Mk$. Now, since $\PER_{M, \eta, L} \in \Pol(\Gamma)$, we have that
  \[
    Y := \PER_{M, \eta, L}(\underbrace{X^1, \hdots, X^1}_{w_1\text{ copies}}, \hdots, \underbrace{X^k, \hdots, X^k}_{w_k\text{ copies}}) \in Q_i.
  \]
  For each coordinate $i \in \{1, \hdots, \ar_i\}$, we have that by definition of $\PER$,
  \[Y_i = \eta\left(\sum_{j=1}^k w_jX^j_i\!\!\!\!\!\mod M\right) = \eta\left(\sum_{j=1}^k r_j(rX^j_i)\!\!\!\!\!\mod M\right) = h\left(\sum_{j=1}^k r_jg(X^j)\right) = h(V_i).\]
  Since $h(V) = Y \in Q_i$, we know that $h$ is a homomorphism from $\LE_{\Gamma}$ to $\Gamma_Q$, as desired.
  
  Since $R$ is a finite commutative ring, we have that $R$ is LE-solvable. Thus, by Theorem~\ref{thm:affine}, we have that $\PCSP(\Gamma) \in \mathsf{P}^h$. Since $h = \eta$ is a constant-sized function, $\PCSP(\Gamma) \in \mathsf{P}$, as desired.
\end{proof}

\subsection{Threshold-periodic Polymorphisms}\label{subsec:threshold-periodic}

It turns out that threshold polymorphisms and periodic polymorphisms can be combined in nontrivial ways

\begin{df}
  Let $T = \{\tau_0 = 0, \tau_1, \hdots, \tau_k = 1\}$ be a threshold sequence, $M = (M_0, \hdots, M_k)$ be a sequence of positive integers, and $H = (\eta_1, \hdots, \eta_{k})$ be a sequence of maps $\eta_i : \mathbb Z/M_i\mathbb Z \to E$. Let $L$ be a positive integer such that $L\tau_i$ is not an integer for any $i \in \{1, \hdots, k - 1\}$. Then, define $\THRPER_{T, M, H, L} : \{0, 1\}^L \to E$ to be the following polymorphism.
  \[\THRPER_{T, M, H, L}(x) = \begin{cases}
      \eta_1(0) & \Ham(x) = 0\\
      \eta_i(\Ham(x)\!\!\!\!\!\!\mod M_i) & L\tau_{i-1} < \Ham(x) < L\tau_i, 1 \le i \le k\\
      \eta_k(L) & \Ham(x) = L.
  \end{cases}\]
\end{df}

For technical reasons, we have to have that values at Hamming weights $0$ and $L$ be consistent with the periodic patterns in the intervals $(0, \tau_1)$ and $(\tau_{k-1}, 1)$, respectively.

\begin{thm}\label{thm:boolean-threshold-periodic}
  Let $T, M, H$ be defined as above.  Let $\Gamma$ be a promise template on the Boolean domain such that $\THRPER_{T, M, H, L} \in \Pol(\Gamma)$ for infinitely many $L$. Then, $\PCSP(\Gamma) \in \mathsf{P}$. 
\end{thm}

\begin{proof}
  Let $M_{\lcm} = \lcm(M_0, \hdots, M_k)$. Like in the periodic case, there must be some $r \in \mathbb Z/M_{\lcm}\mathbb Z$ such that $L \equiv r \mod M_{\lcm}$ for infinitely many $L$ for which $\THRPER_{T, M, H, L} \in \Pol(\Gamma)$. 

  Pick an LP-solvable ring $A$ such that $\tau_i \not\in A$ for all $i \in \{1, \hdots, k-1\}$. Let $R = \mathbb Z/M_{\lcm}\mathbb Z$. We claim that  $\Gamma$ sandwiches $\LPLE_{A, R}$ via $(g, h)$ where
  \begin{align*}
    g(0) &= (0, 0) \in A \times R\\
    g(1) &= (1, r) \in A \times R\\
    h(x, y) &= \begin{cases}
      \eta_1(y\!\!\!\!\!\!\mod M_0) & x \le 0\\
      \eta_i(y\!\!\!\!\!\!\mod M_i) & \tau_{i-1} < x < \tau_i, 1 \le i \le k\\
      \eta_k(y\!\!\!\!\!\!\mod M_k) & x \ge 1
    \end{cases}
  \end{align*}  
  The justification of this sandwich is a merging of the methods of Theorem~\ref{thm:boolean-threshold} and Theorem~\ref{thm:boolean-periodic}.     Since we desire "access" to each coordinate of $g$, we let $g_A$ be the first coordinate and $g_R$ be the second coordinate.

  Consider $\LPLE_{\Gamma} := \{R_i := (\Conv_A(g_A(P_i)_1), \Aff_R(g_R(P_i))) : P_i \in \Gamma_P\}$ note that $\LPLE_{\Gamma} \subset \LPLE_{A, R}$. By design, $g$ is a homomorphism from $\Gamma_P$ to $\LPLE_{\Gamma}$. Thus, it suffices to show that for any $(V, W) \in R_i$, we have that $h(V, W) \in Q_i$.

  By definition, if $(V, W) \in R_i$, there exists $X^1, \hdots, X^m \in P_i$ as well as $\alpha_1, \hdots, \alpha_m \in [0, 1]$ summing to $1$ and $r_1, \hdots, r_m \in R$ summing to $1$ such that\footnote{The reason these can be simultaneously true is that we can set some $\alpha_i$'s and $r_i$'s to $0$.}
  \begin{align*}
    V &= \alpha_1g_A(X^1) + \cdots + \alpha_mg_A(X^m)\\
    W &= r_1g_R(X^1) + \cdots + r_mg_R(X^m).
  \end{align*}
  Pick $L$ sufficiently larger (to be specified) such that $\THRPER_{T, M, H, L} \in \Pol(\Gamma)$ with $L \equiv r \mod M_{\lcm}$. We now need to delicately find integer weights $w_1, \hdots, w_m$ such that the following properties hold
  \begin{align*}
     \sum_{i=1}^m w_i &= L\\
    w_i &\equiv r_ir\mod M\text{ for all }i\\
    \left|w_i - \alpha_i L\right| &\le M\text{ for all }i.
  \end{align*}
  Note that the first two conditions are consistent because $\sum_{i=1}^m r_ir \equiv r \equiv L \mod M$. Such $w_i$'s can be constructed by first setting each $w_i$ to be the greatest integer at most $\alpha_i L$ which is equivalent to $r_ir \mod M$. Then, one can increase the $w_i$'s by $M$ one-by-one until they sum to $L$.

  With these in hand, consider
  \[Y := \THRPER_{T, M, H, L}(\underbrace{X^1, \hdots, X^1}_{w_1\text{ copies}}, \hdots, \underbrace{X^m, \hdots, X^m}_{w_m\text{ copies}}) \in Q_i.\]
  Define for each coordinate $j \in \{1, \hdots, \ar_i\}$
  \begin{align*}
    s_j^A &:= \frac{1}{L}\Ham(\underbrace{X^1_j, \hdots, X^1_j}_{w_1\text{ copies}}, \hdots, \underbrace{X^m_j, \hdots, X^m_j}_{w_m\text{ copies}})\\
    s_j^R &:= \sum_{a=1}^m w_aX^a_j \mod M
  \end{align*}
  Then, by design, $Y_j = h(s_j^A, s_j^R)$. We also know that
  \[
    \frac{1}{L}|s_j^A - V_j|= \sum_{\substack{a = 1\\X^a_j = 1}}^k \left|\frac{w_a}{L} - \alpha_a\right| \le \frac{Mm}{L}
  \]
  as well as
  \[
    s_j^R = \sum_{a=1}^m r_j(rX^j_i) \mod M = \sum_{a=1}^m r_jg_R(X^j_i) = W_j.
  \]
  Since $V \in A^{\ar_i}$, we have that $\tau_{i-1} < V_j < \tau_i$ for $i \in \{2, \hdots, k-1\}$ or $\tau_0 \le V_j < \tau_1$ or $\tau_{k-1} < V_j \le \tau_k$. In any case, as $L \to \infty$, $s_j^A$ will get sufficiently close to $V_j$ so that it falls into the same interval as $V_j$. Thus, $Y_j = h(s_j^A, s_j^R) = h(V_j, W_j)$. Therefore, $h(V, W) = Y \in Q_i$, establishing the sandwiching is valid.
  
  Since $A$ is LP-solvable and $R$ is LE-solvable, by Theorem~\ref{thm:basiclp+affine}, we have that $\PCSP(\Gamma) \in \mathsf{P}^h$. Since $h$ only needs to check thresholds and then use a finite lookup table, $h$ can be computed in polynomial time in the description of the input. Thus, $\PCSP(\Gamma) \in \mathsf{P}$, as desired.
\end{proof}

\section{Regional Boolean polymorphisms}\label{sec:Boolean-block}

So far, all of the families of polymorphisms we have studied are Boolean, symmetric; that is, they only depend on the Hamming weight of the input vector. This section describes how these results can be extended to special kinds of \emph{block symmetric} functions. Like in the previous section, our promise domain is always $(D = \{0, 1\}, E, \phi)$. 

\begin{df}
  Let $b$ and $L$ be positive integers. A function $f : D^L \to E$ is $b$-block symmetric, if there is a partition $[L] = B_1 \cup B_2 \cup \cdots \cup B_b$ such that for all $(x_1, \hdots, x_L) \in D^L$ and any permutation $\pi : [L] \to [L]$ such that $\pi(B_i) = B_i$ for all $i$.
  \[
    f(x_1, \hdots, x_L) = f(x_{\pi(1)}, \hdots, x_{\pi(L)}).
  \]
\end{df}

In other words, $f$ is $b$-block symmetric with the corresponding partition $B_1 \cup \cdots \cup B_b$, then, $f(x)$ depends only on $(\Ham_{B_1}(x), \hdots, \Ham_{B_b}(x))$, where $\Ham_{B_i}(x)$ is the sum of the coordinates with indices in $B_i$. Analogous to how a symmetric function can be thought of as a function on the real interval $[0, 1]$, a $b$-block symmetric function can be thought of as a function on $[0, 1]^b$.

\subsection{Regional Polymorphisms}\label{subsec:opwp}

\newcommand{\Part}{\operatorname{Part}}

Even going from $[0, 1]$ to $[0, 1]^2$, the ways of splitting up space can become rather complex. Thus, instead of giving an explicit description like for threshold polymorphisms, we discuss a generalization which we call \emph{open partition polymorphisms}. First, we need to define what an \emph{open partition} is.

\begin{df}
  Let $A_1, \hdots, A_b \subset \mathbb R$ be dense commutative rings. Let $\mathfrak A := (A_1 \times A_2 \times \cdots \times A_b) \cap [0, 1]^b$. Let $E$ be a set. A function $\Part : \mathfrak A \to E$ is an \emph{open partition} if for all $x \in \mathfrak A$, there exists $\epsilon > 0$, such that for all $y \in \mathfrak A$ with $|x - y| < \epsilon$, we have $\Part(y) = \Part(x)$. In other words, for all $e \in E$, $\Part^{-1}(e)$is  open in the Euclidean topology induced by $\mathfrak A$.
\end{df}

We also have a slightly more general notion called an \emph{integer open partition} which allows for arbitrary values to be set at the corners of the hypercube $[0, 1]^b$.

\begin{df}
  Let $A_1, \hdots, A_b \subset \mathbb R$ be dense commutative rings. Let $\mathfrak A := (A_1 \times A_2 \times \cdots \times A_b) \cap [0, 1]^b$. Let $E$ be a set. A function $\Part : \mathfrak A \to E$ is an \emph{integer open partition} if for all $x \in \mathfrak A \setminus \{0, 1\}^b$, there exists $\epsilon > 0$, such that for all $y \in \mathfrak A$ with $|x - y| < \epsilon$, we have $\Part(y) = \Part(x)$. In other words,for all $e \in E$,  $\Part^{-1}(e) \setminus \mathbb \{0, 1\}^b$ is open in the Euclidean topology induced by $\mathfrak A$.
\end{df}

Going back to the $1$-dimensional case, consider $A_1 = \mathbb Z[\sqrt{2}]$ so that $\mathfrak A = A_1 \cap [0, 1]$. Also, let our range be $E = \{0, 1\}$. The partition corresponding to the $\MAJ$ polymorphism is then
\[
  \Part_{\MAJ}(x) = \begin{cases}
    0 & x < 1/2\\
    1 & x > 1/2.
  \end{cases}
\]
This function is an open partition because the apparent boundary element $1/2$ does not exist in $\mathbb Z[\sqrt{2}]$. On the other hand, the partition corresponding to the $\AND$ polymorphism.
\[
  \Part_{\AND}(x) = \begin{cases}
    0 & x < 1\\
    1 & x = 1
  \end{cases}
\]
is not an open partition because $\Part^{-1}(1)$ has boundary. Yet, it is an integer open partition because the only boundary term has integer coordinates.

A more complex example in two dimensions is as follows. Let $\mathfrak A = (\mathbb Z[\sqrt{2}] \times \mathbb Z[\sqrt{3}]) \cap [0, 1]^2$ and let
\[
  \Part_{\AT}(x, y) = \begin{cases}
    0 & x < y\\
    1 & x > y\\
    0 & x = y = 0\\
    1 & x = y = 1
  \end{cases}
\]
See Figure~\ref{fig:region-example}. Note that since the two coordinates are in the rings $\mathbb Z[\sqrt{2}]$ and $\mathbb Z[\sqrt{3}]$, $x = y$ if and only if $x$ and $y$ are both integers. Thus, the only boundary terms have integer coordinates, so $\Part_{\AT}$ is an integer open partition. As hinted by the name, $\Part_{\AT}$ is connected to the family of polymorphisms $\AT_L$. This connection is made more explicit soon.

For a more nontrivial example, consider $E = \{0, 1, 2, 3, 4\}$ and $\mathfrak A = (\mathbb Z[\sqrt{2})^2 \cap [0, 1]^2$. Let
\[
  \Part_{\text{circle}}(x, y) = \begin{cases}
    0 & x < 1/2\text{ and }y < 1/2\text{ and }(x - 1/2)^2 + (y - 1/2)^2 > 1/13\\
    1 & x < 1/2\text{ and }y > 1/2\text{ and }(x - 1/2)^2 + (y - 1/2)^2 > 1/13\\
    2 & x > 1/2\text{ and }y < 1/2\text{ and }(x - 1/2)^2 + (y - 1/2)^2 > 1/13\\
    3 & x > 1/2\text{ and }y > 1/2\text{ and }(x - 1/2)^2 + (y - 1/2)^2 > 1/13\\
    4 & (x - 1/2)^2 + (y - 1/2)^2 < 1/13\\
  \end{cases}
\]

In this case, $\Part_{\text{circle}}$ is an open partition, since the equations $ x= 1/2$, $y = 1/2$ and $(x - 1/2)^2 + (y - 1/2)^2 = 1/13$ have no solutions in $(\mathbb Z[\sqrt{2}])^2$.

\begin{center}
  \begin{figure}
    \begin{center}
      \includegraphics[width=1.7in]{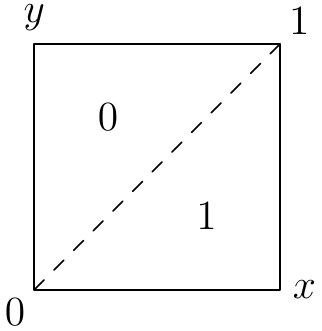}\hspace{1.5in}
      \includegraphics[width=1.7in]{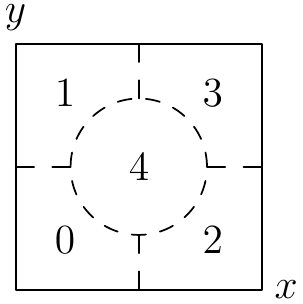}
    \end{center}
    \caption{Plots of $\Part_{\AT}(x, y)$ and $\Part_{\text{circle}}(x, y).$ The dashed lines represent the boundary between the regions. The $0$ and $1$ in the corners of the square for $\Part_{\AT}(x, y)$ represents the value chosen at those corners.}
    \label{fig:region-example}
  \end{figure}
\end{center}

Although an integer open partition $\Part$ is only defined in $\mathfrak A$, we can extend it to a substantial portion of $[0, 1]^b$. This is useful when we desire to discretize $\Part$ by wanting know its value at particular rational coordinates (which may not be in $\mathfrak A$).

\begin{df}
  Let $\Part : \mathfrak A \to E$ be an integer open partition. Define $\overline{\Part} : [0,1]^b \to E\cup \{\perp\}$ to be the partial function for which
  \[
    \overline{\Part}(x) = \begin{cases}
      \Part(x) & x \in \{0, 1\}^n\\
      e \in E & \exists \epsilon > 0, \forall y \in \mathfrak A, |x - y| < \epsilon\text{ implies } \Part(x) = e\\
      \perp & \text{otherwise}.
    \end{cases}
  \]
\end{df}
Note that since $\Part$ is an integer open partition, $\Part(x) = \overline{\Part}(x)$ for all $x \in \mathfrak A$. Since $\mathfrak A$ is dense in $[0, 1]^b$, and $\mathfrak A \subset (\overline{\Part})^{-1}(E)$ is open, we have that $(\overline{\Part})^{-1}(\perp) = [0, 1]^b \setminus \mathbb (\overline{\Part})^{-1}(E)$ is nowhere dense, although it may have positive Lebesgue measure.

As alluded to earlier, these partitions $\Part$, which will end up being our rounding functions, are discretized to form a collection of polymorphisms. Recall our domain $D = \{0, 1\}$ is Boolean for this section.

\begin{df}
  Let $A_1, \hdots, A_b \subset \mathbb R$ be dense commutative rings. Let $\mathfrak A = A_1 \times \cdots \times A_b \cap [0, 1]^b$. Let $\Part : \mathfrak A \to E$ be an integer open partition. Let $L_1, \hdots, L_b$ be positive integers such that for all $k_i \in \{0, 1, \hdots, L_i\}$ for all $k_i \in [b]$, we have that $\overline{\Part}\left(\frac{k_1}{L_1}, \hdots, \frac{k_b}{L_b}\right) \neq \perp$. Let $L = \sum_{i=1}^b L_i$ and let $\mathcal B = (B_1, \hdots, B_b)$ be a partition of $[L]$ such that $|B_i| = L_i$ for all $i \in \{1, \hdots, b\}$. Define the \emph{regional polymorphism} $\REG_{\Part, \mathcal B} : D^{B_1} \times \cdots \times D^{B_b} \to E$ to be
  \[
    \REG_{\Part, \mathcal B}(x) = \overline{\Part}\left(\frac{\Ham_{B_1}(x)}{L_1}, \cdots \frac{\Ham_{B_b}(x)}{L_b}\right).
  \]
\end{df}

For example, $\REG_{\Part_{\AT}, (\{1, 3, \hdots, 2k+1\}, \{2, 4, \hdots, 2k\})}$ is the same as $\AT_{2k+1}$ up to a permutation of the coordinates.

Now, we can prove that having an infinite collection of regional polymorphisms implies tractability as long as $\Part$ is efficiently computable.

\begin{thm}\label{thm:regional}
  Let $A_1, \hdots, A_b \subset \mathbb R$ be LP-solvable subrings. Let $\mathfrak A = A_1 \times \cdots \times A_b \cap [0, 1]^b$. Let $\Part : \mathfrak A \to E$ be an integer open partition. Let $\Gamma = (\Gamma_P = \{P_i \in D^{\ar_i} : i \in I\}, \Gamma_Q = \{Q_i \in E^{\ar_i}\})$ be a promise template. Assume that for all positive integers $\ell$, there exists $\REG_{\Part, \mathcal B} \in \Pol(\Gamma)$ such that $|B_i| \ge \ell$ for all $B_i \in \mathcal B$. Then, $\PCSP(\Gamma) \in \mathsf{P}^{\Part}$. 
\end{thm}

\begin{proof}
  Let $\mathcal A = (A_1, \hdots, A_b)$. By Theorem~\ref{thm:basiclp+affine}, it suffices to show that $\Gamma$ sandwiches $\LPLE_{\mathcal A}$. The map $g : D \to \mathfrak A$ is just
  \[
    g(d) = \begin{cases}
      (0, \hdots, 0) & d = 0\\
      (1, \hdots, 1) & d = 1.
    \end{cases}
  \]
  As suggested by the theorem statement, the rounding map is precisely\footnote{Technically, the domain of $h$ is $A_1 \times \cdots \times A_b$, whereas the domain of $\Part$ is $[0, 1]^b$. This can be ``fixed'' by having $h$ return a default value (e.g., $0$) when the input is outside $[0, 1]^b$.} the integer open partition $h = \Part$.

  Now, let $\LPLE_{\Gamma} = \{R_i := \Conv_{A_1}(g_1(P_i)) \times \cdots \times \Conv_{A_b}(g_b(P_i)) \in \mathfrak A^{\ar_i}: i \in I\}$. By design, $g$ is a homomorphism from $\Gamma_P$ to $\LPLE_{\Gamma}$. The heart of the argument is to show that $h$ is a polymorphism from $\LPLE_{\Gamma}$ to $\Gamma_Q$. In other words, we need to show for all $i \in I$, that $h(R_i) \subseteq Q_i$. Fix $V \in R_i$ and view $V = (V_1, \hdots, V_b) \in A_1^{\ar_i} \times \cdots \times A_b^{\ar_i}$. With $V_a = (V_{a, 1}, \hdots, V_{a, \ar_i}) \in A_a^{\ar_i}$.

  List the elements $X^1, \hdots, X^m \in P_i$. For all $a \in \{1, 2, \hdots, b\}$, because $V_a \in \Conv_{A_a}(g_1(P_i))$, we have that there exists weights $\alpha_{a,j} \in [0, 1]$ with $j \in \{1, \hdots, m\}$ such that
  \[
    V_a = \sum_{j=1}^m \alpha_{a,j}X^j.
  \]
  (Note that $g$ can be omitted, since it is the identity map on each coordinate.) Pick $\ell$ sufficiently large (to be determine later), such that $\REG_{h, \mathcal B} \in \Pol(\Gamma)$ and $|B_a| \ge \ell$ for all $B_a \in \mathcal B$. Then, using a nearly identical argument as the one in Theorem~\ref{thm:boolean-threshold}, we can find integer weights $w_{a, j}$ such that $\sum_{j=1}^m w_{a, j} = |B_a|$ for all $a \in \{1, 2, \hdots, b\}$ and
  \[
    \left|\frac{w_{a, j}}{|B_a|} - \alpha_{a,j}\right| \le \frac{1}{|B_a|} \le \frac{1}{\ell}.
  \]
  Furthermore, we can ensure that $w_{a,j} = 0$ whenever $\alpha_{a,j} = 0$
  Fix $k \in \{1, \hdots, \ar_i\}$. There are essentially two cases to consider
  \begin{itemize}
  \item If $W_k := (V_{1,k}, \hdots, V_{b,k}) \in [0, 1]^b \setminus \{0, 1\}^b$, consider $\epsilon > 0$ such that $\overline{\Part}(x) = \overline{\Part}(W_k)$ for all $|x - W_k| < \epsilon$. Then, if $\ell$ is chosen such that $\frac{b}{\ell} < \epsilon$, then for each $k \in \{1, \hdots, \ar_i\}$
    \begin{align*}
      \REG_{\Part,\mathcal B}(\underbrace{X^1_k, \hdots, X^1_k}_{w_{1, 1}\text{ copies}}, \hdots, \underbrace{X^m_k, \hdots, X^m_k}_{w_{1,m}\text{ copies}},&
                              \underbrace{X^1, \hdots, X^1}_{w_{2,1}\text{ copies}}, \hdots, \underbrace{X^m_k, \hdots, X^m_k}_{w_{2, m}\text{ copies}},
                              \hdots
                               \underbrace{X^1_k, \hdots, X^1_k}_{w_{b,1}\text{ copies}}, \hdots, \underbrace{X^m_k, \hdots, X^m_k}_{w_{b, m}\text{ copies}})\\
                                                                                                                                                            &= \overline{\Part}\left(\frac{\sum_{j=1}^m w_{1,j}X^j_k}{|B_1|}, \hdots, \frac{\sum_{j=1}^m w_{b,j}X^j_k}{|B_b|}\right)\\
                                                                                                                                                            &= \overline{\Part}\left(\sum_{j=1}^m \alpha_{1,j}X^j_k, \hdots, \sum_{j=1}^m \alpha_{b,j}X^j_k\right)\text{ (within $\epsilon$)}\\
                                                                                                                                                            &= \overline{\Part}(W_k)\\
                                                                                                                                                            &= \Part(W_k).
    \end{align*}
  \item Otherwise, if $W_k \in \{0, 1\}^b$, whenever $\alpha_{a,j} \neq 0$, we must have that $V_{a,k} = X^j_k$. Since $\alpha_{a,j} = 0$ implies $w_{a,j} = 0$, we have that
    \begin{align*}
      \REG_{\Part,\mathcal B}(&\text{same as above})\\
      &= \overline{\Part}\left(\frac{\sum_{j=1}^m w_{1,j}X^j_k}{|B_1|}, \hdots, \frac{\sum_{j=1}^m w_{1,j}}{|B_b|}\right)\\
      &= \overline{\Part}\left(\frac{\sum_{j=1}^m w_{1,j}V^j_k}{|B_1|}, \hdots, \frac{\sum_{j=1}^m w_{1,j}}{|B_b|}\right)\\
      &= \overline{\Part}(W_k)\\
      &= \Part(W_k)\text{ (because $W_k \in \{0, 1\}^b$}.
    \end{align*}
  \end{itemize}
  In either case, we have that $\Part(V)$ is the output of $\REG_{\Part,\mathcal B}(P_i) \subseteq Q_i.$ Thus, we have the aforementioned sandwich.
\end{proof}

\subsection{Regional-periodic Polymorphisms}\label{subsec:reg-per}

Just as threshold polymorphisms can be generalized to threshold-periodic polymorphisms, we have that regional polymorphisms can be generalized to \emph{regional-periodic} polymorphisms.

Recall that if $\mathfrak A := A_1 \times \cdots \times A_b \cap [0, 1]$, where the $A_i$'s are dense commutative rings, then $\Part : \mathfrak A \to E$ is an open partition if for all $e \in E$, $\Part^{-1}(e)$ is relatively open with respect to the Euclidean topology induced by $\mathfrak A$. In other words, for all $e \in E$, there exists $\Omega_e \subset \mathbb R^b$ open such that $\Part^{-1}(e) = \Omega_e \cap \mathfrak A$. We call $\Omega_e$ a \emph{region} of $\Part$. Note that for all $x \in \Omega_e \cap [0, 1]^b$, $\overline{\Part}(x) = e$. Given this, we can now define \emph{regional-periodic polymorphisms}.

\begin{df}
  Let $A_1, \hdots, A_b \subset \mathbb R$ be dense commutative rings. Let $\mathfrak A = A_1 \times \cdots \times A_b$ be a product of subrings of $\mathbb R$. Let $S$ be a finite set, and let $\Part : \mathfrak A \to S$ be an open partition. Let $L_1, \hdots, L_b$ be positive integers such that for all $k_i \in \{0, 1, \hdots, L_i\}$ for all $i \in [b]$, we have that $\overline{\Part}\left(\frac{k_1}{L_1}, \hdots, \frac{k_b}{L_b}\right) \neq \perp$. Let $L = \sum_{i=1}^b L_i$ and let $\mathcal B = (B_1, \hdots, B_b)$ be a partition of $[L]$ such that $|B_i| = L_i$ for all $i \in [b]$. For each $k \in S$, let $J_k$ be an ideal of $\mathbb Z^b$ such that $\mathbb Z^b/J_k$ is finite. Let $\mathcal M  = \{M_k: \mathbb Z^b/J_k \to E \mid k \in S\}$ be a collection of maps. Define the \emph{regional-periodic polymorphism} $\REGPER_{\Part, \mathcal B, \mathcal M} : D^{B_1} \times \cdots \times D^{B_b} \to E$ to be
  \[
    \REGPER_{\Part, \mathcal B, \mathcal M}(x) = M_k((\Ham_{B_1}(x), \hdots, \Ham_{B_b}(x)) \!\!\!\!\!\!\!\mod J_k)\text{ where }k = \overline{\Part}\left(\frac{\Ham_{B_1}(x)}{L_1}, \cdots \frac{\Ham_{B_b}(x)}{L_b}\right).
  \]
\end{df}

Figure~\ref{fig:regper-example} shows an example of a partition with periodic functions added.  Now, we can prove a more general result.

\begin{center}
  \begin{figure}
    \begin{center}
      \includegraphics[width=2.2in]{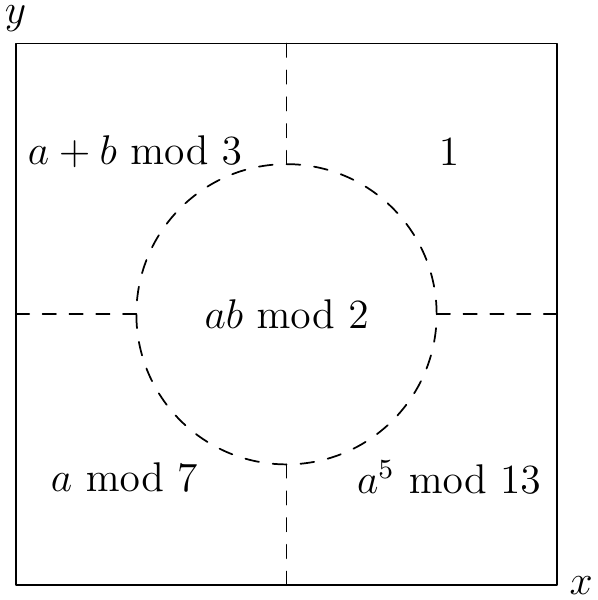}
    \end{center}
    \caption{Plot of $\Part_{\text{circle}}(x, y)$ in the same style as Figure~\ref{fig:region-example}. Within each region is a function $(a, b) \mapsto M_i(a, b)$ whose domain is some quotient of $\mathbb Z^2$, where $a$ and $b$ represent the Hamming weights of each block in the corresponding regional-periodic polymorphism.}
    \label{fig:regper-example}
  \end{figure}
\end{center}

\begin{thm}\label{thm:regional-periodic}
  Let $A_1, \hdots, A_b \subset \mathbb R$ be LP-solvable rings. Let $\mathfrak A = A_1 \times \cdots \times A_b \cap [0, 1]^b$. Let $\Part : \mathfrak A \to S$ be an open partition. Let $\mathcal M  = \{M_k : \mathbb Z^b/J_k \to E \mid k \in S\}$ be a collection of maps. Let $\Gamma = (\Gamma_P = \{P_i \in D^{\ar_i} : i \in I\}, \Gamma_Q = \{Q_i \in E^{\ar_i}\})$ be a promise template. Assume that for all positive integers $\ell$, there exists $\REGPER_{\Part, \mathcal B, \mathcal M} \in \Pol(\Gamma)$ such that $|B_i| \ge \ell$ for all $B_i \in \mathcal B$. Then, $\PCSP(\Gamma) \in \mathsf{P}^{\Part}$. 
\end{thm}

The proof has similar structure to the proof of the threshold-periodic case, Theorem~\ref{thm:boolean-threshold-periodic}.

\begin{proof}
  Given a sequence of blocks $\mathcal B$, we can define its \emph{residue} with respect to $\mathcal M$ to be the sequence
  \[
    \Res_{\mathcal M}(\mathcal B) = (\mathcal B \mod J_k : k \in S).
  \]
  Note that since $\mathbb Z^b/J_k$ is a finite quotient for all $k \in S$, the set of all possible residues is finite. Thus, there exists a residue $\hat{r} := (\hat{r}_k : k \in S)$ such that $\hat{r} = \Res_{\mathcal M}(\mathcal B)$ for infinitely many $\mathcal B$ such that $\REGPER_{\Part, \mathcal B, \mathcal M} \in \Pol(\Gamma)$ and $\min\{|B_i|\}$ is arbitrarily large.

  We now apply the ring-theoretic Chinese Remainder Theorem to our quotient rings. Let $J = \bigcap_{k \in S} J_k$ be the intersection of the ideals of $\mathbb Z^b$. Note that $J$ is also an ideal of $\mathbb Z^b$. Furthermore, $\mathbb Z^b / J$ is finite, as any two elements of $x, y \in \mathbb Z^b$ with the same residue satisfy $x - y \in J_k$ for all $k \in S$, so $x - y \in J$, implying a finite number of cosets. This also implies that we can identify $\hat{r}$ with an element of $\mathbb Z^b/J$.

  Now, let $\mathcal A = (A_1, \hdots, A_b)$ and $\mathcal R = (\mathbb Z^b/J)$. Recall that $\LPLE_{\mathcal A, \mathcal R}$ is the direct product
  \[
    \LPLE_{\mathcal A, \mathcal R} = \left(\bigtimes_{j=1}^b \LP_{A_j}\right) \times \LE_{\mathbb Z^b/J}.
  \]
  We claim that $\Gamma$ sandwiches $\LPLE_{\mathcal A, \mathcal R}$ via the maps $(g, h)$ where
  \begin{align}
    g(0) &= (\underbrace{0, \hdots, 0}_{b\text{ terms}}, 0)\nonumber\\
    g(1) &= (\underbrace{1, \hdots, 1}_{b\text{ terms}}, \hat{r})\nonumber\\
    h(x_1, \hdots, x_b, r) &= M_k(r)\text{ where $k = \Part(x_1, \hdots, x_b)$.}\label{eq:h}
  \end{align}

  Analogous to the proof of Theorem~\ref{thm:boolean-threshold-periodic}, define
  \[
    \LPLE_{\Gamma} := \{R_i := (\Conv_{A_1}(g_{1}(P_i)), \hdots, \Conv_{A_b}(g_{b}(P_i)), \Aff_{\mathbb Z^b/J}(g_{b+1}(P_i))): P_i \in \Gamma_P\}.
  \]
  As before, by definition, $g$ is a homomorphism from $\Gamma_P$ to $\LPLE_{\Gamma}$. Thus, it suffices to show
  \[\text{for all }(V_1, \hdots, V_b, W) \in R_i,\text{ we have that }h(V_1, \hdots, V_b, W) \in Q_i.\]
  
  Let $X^1, \hdots, X^m$ be the elements of $P_i$.  For all $j \in [b]$, since $V_j  \in \Conv_{A_j}(g_j(P_i))$, we have that there exists $\alpha_{j, 1}, \hdots, \alpha_{j, m} \in [0, 1]$ summing to $1$ such that 
  \[
    V_j = \alpha_{j,1}g_{a}(X^1) + \cdots + \alpha_{j,m}g_{a}(X^m).
  \]
  Likewise, since $W \in \Aff_{\mathbb Z^b/J}(g_{b+1}(P_i))$, we have that there exist $r_1, \hdots, r_m \in \mathbb Z^b/J$ which sum to the identity $(1, \hdots, 1) \in \mathbb Z^b/J$ such that
  \[
    W = r_1g_{b+1}(X^1) + \cdots + r_mg_{b+1}(X^m).
  \]
  Fix $\ell$ sufficiently large (to be specified later) and $\mathcal B = (B_1, \hdots, B_b)$ with $|B_j| \ge \ell$ for all $j \in [b]$ such that $\Res_{\mathcal M}(\mathcal B) = \hat{r} \in \mathbb Z^b/J$. For any such $j \in [b]$ find weights $w_{j,1}, \hdots, w_{j,m}$ satisfying the following conditions:
  \begin{align*}
    \sum_{k=1}^m w_{j,k} = |B_j|&\text{ for all $j \in [b]$}&\text{(cardinality condition)}\\
    (w_{1,k}, \hdots, w_{b,k}) \in r_k\hat{r} + J&\text{ for all $k \in [m]$}&\text{(coset condition)}\\
    \left|w_{j,k} - \alpha_{j,k}|B_{j,k}|\right| \le 2|\mathbb Z^b/J|bm&\text{ for all $j \in [b], k \in [m].$} &\text{(approximation condition)}
  \end{align*}
  This can be done by first estimating $\hat{w}_{j,k} = \lfloor \alpha_{j,k}|B_j|\rfloor$ and then adjusting each as little as possible (at most $|\mathbb Z^b/J|$)  so that $(\hat{w}_{1,k}, \hdots, \hat{w}_{b,k})$ are in the appropriate cosets. Then, it is not hard to check by the compatibility of the conditions that
    \[
      T := (|B_1|, \hdots, |B_b|) - \left(\sum_{k=1}^m w_{1,k}, \hdots, \sum_{k=1}^m w_{b,k}\right) \in J.
    \]
    If we add $T$ to $(\hat{w}_{1,1}, \hdots, \hat{w}_{b,1})$, then the cardinality and coset constraints are satisfied.  Note that the entries of $T$ are bounded by $|\mathbb Z^b/J|m$, so the approximation condition is also still satisfied.

  Now consider
  \begin{align*}
    Y := \REGPER_{\Part,\mathcal B, \mathcal M}(&\underbrace{X^1, \hdots, X^1}_{w_{1,1}\text{ copies}}, \hdots, \underbrace{X^m, \hdots, X^m}_{w_{1,m}\text{ copies}},\\
    &\underbrace{X^1, \hdots, X^1}_{w_{2,1}\text{ copies}}, \hdots, \underbrace{X^m, \hdots, X^m}_{w_{2,m}\text{ copies}}, \hdots\\
    &\underbrace{X^1, \hdots, X^1}_{w_{b,1}\text{ copies}}, \hdots, \underbrace{X^m, \hdots, X^m}_{w_{b,m}\text{ copies}}) \in Q_i.
  \end{align*}

  We seek to prove that $h(V_1, \hdots, V_b, W) = Y$, showing that $h(V_1, \hdots, V_b, W) \in Q_i$.

  For each $j \in [b]$ and each coordinate $k \in \{1, \hdots, \ar_i\}$ (recall $\ar_i$ is the arity of $P_i$, $Q_i$ and $R_i$) we can define
  \begin{align*}
    s^{A_j}_k &:= \frac{1}{|B_j|}\Ham(\underbrace{X_k^1, \hdots, X_k^1}_{w_{j,1}\text{ copies}}, \hdots, \underbrace{X_k^m, \hdots, X_k^m}_{w_{j,m}\text{ copies}}).\\
              &= \frac{1}{|B_j|}\sum_{\beta = 1}^m w_{j,\beta}X_k^\beta\\
              &\approx \sum_{\beta = 1}^m \alpha_{j,\beta}X_k^\beta = V_{j,k},
  \end{align*}
  where $\approx$ means $O(\frac{1}{\ell})$ error. Thus, if $\ell$ is sufficiently large, for all $k \in [\ar_i]$, $(s^{A_1}_k, \hdots, s^{A_b}_k)$ will be in the same region of $\Part$ as $(V_{1,k}, \hdots, V_{b,k})$ because the regions are relatively open.

  Furthermore, for all $k \in [\ar_i]$ define
  \begin{align*}
    s^{\mathcal R}_k &= \sum_{\beta=1}^m (w_{1,\beta}, \hdots, w_{b,\beta})X^\beta_k\\
                     &\in J + \sum_{\beta=1}^m r_\beta(\hat r X^\beta_k)\\
                     &= J + \sum_{\beta =1}^m r_\beta g_{b+1}(X^\beta_k) = W_k.
  \end{align*}

  Thus, for all $k \in [\ar_i]$,
  \begin{align*}
    Y_k &= M_{\overline{\Part}(s^{A_1}_k, \hdots, s^{A_b}_k)}(s^{\mathcal R}_k)\\
        &= M_{\Part(V_{1,k}, \hdots, V_{b,k})}(W_k)\text{ (by above discussion)}\\
        &= h_k(V_1, \hdots, V_b, W)\text{ (by (\ref{eq:h}))}.
  \end{align*}
  Thus, $Y = h(V_1, \hdots, V_b, W)$, so we have established the homomorphic sandwich. Since each $A_i$ is LP-solvable, and $\mathbb Z^b/J$ is a finite commutative ring (and so is LE-solvable), by Theorem~\ref{thm:basiclp+affine}, we have that $\PCSP(\Gamma) \in \mathsf{P}^{\Part}$.  
\end{proof}

\section{Extending to Larger Domains}\label{sec:larger}

Given the established framework, the extension from Boolean to non-Boolean domains is not much more difficult. The main change is that instead of relaxing the domain $D$ to the interval $[0, 1]$, we map to the \emph{standard $D$-simplex}:

\[\Delta^D := \{x \in \mathbb R^D : x_d \ge 0\text{ for all }d \in D, \sum_{d \in D} x_d = 1\}.\]

We let $e^d \in \mathbb R^D$ denote the standard basis and further assume without loss of generality that $D = \{1, 2, \hdots, |D|\}$. We also  need to discuss a generalization of Hamming weight. Given a vector $x \in D^L$ and an element $d \in D$, we define
\[
  \Ham_d(x) = |\{i \in [L] : x_i = d\}|.
\]

\newcommand{\HamVec}{\operatorname{HamVec}}

In this section, we let $\HamVec(x)$ denote the vector in $\mathbb Z^{D}$ such that the $d$th coordinate is $\Ham_d(x)$. For example, in the Boolean case (with $D = \{1, 2\}$), $\HamVec(x) = (L - \Ham(x), \Ham(x)).$

\subsection{Simplex polymorphisms}

\newcommand{\SIMPLEX}{\operatorname{SIMPLEX}}

In this section, we immediately generalize regional-periodic polymorphisms (as opposed to starting with just a generalization of regional). This leads directly to the definition of \emph{simplex polymorphisms}. These are very similar to regional-periodic polymorphisms, except they correspond to partitions of simplices instead of hypercubes. One key difference is that since we no longer have the product set structure, we need to use the same subring for every coordinate.

\begin{df}
  Let $A \subset \mathbb R$ be a dense subring. Let $\mathfrak A = A^{D}$, $S$ be a finite set, and $\Part : \mathfrak A \to S$ be an open partition. Let $L$ be a positive integer such that for all nonnegative $\{\ell_d : d \in D\}$ summing to $L$, we have that
  \[\overline{\Part}\left(\frac{\ell_1}{L}, \hdots, \frac{\ell_{|D|}}{L}\right) \neq \perp.\]
  For each $k \in S$, let $J_k$ be an ideal of $\mathbb Z^{D}$ such that $\mathbb Z^{D}/J_k$ is finite. Let $\mathcal M  = \{M_k: \mathbb Z^{D}/J_k \to E \mid k \in S\}$ be a collection of maps. Define the \emph{simplex polymorphism} $\SIMPLEX_{\Part, L, \mathcal M} : D^{L} \to E$ to be
  \[
    \SIMPLEX_{\Part, L, \mathcal M}(x) = M_k(\HamVec(x) \!\!\!\!\!\!\!\mod J_k)\text{ where }k = \overline{\Part}\left(\frac{\HamVec(x)}{L}\right).
  \]
\end{df}

\begin{thm}
  Let $A \subset \mathbb R$ be an LP-solvable ring. Let $\mathfrak A = A^D$. Let $\Part : \mathfrak A \to S$ be an open partition. Let $\mathcal M  = \{M_k : \mathbb Z^D/J_k \to E \mid k \in S\}$ be a collection of maps. Let $\Gamma = (\Gamma_P = \{P_i \in D^{\ar_i} : i \in I\}, \Gamma_Q = \{Q_i \in E^{\ar_i}\})$ be a promise template. Assume that there exists $\SIMPLEX_{\Part, L, \mathcal M} \in \Pol(\Gamma)$ for arbitrarily large $L$. Then, $\PCSP(\Gamma) \in \mathsf{P}^{\Part}$. 
\end{thm}

The proof has similar structure to the proofs of Theorem~\ref{thm:boolean-threshold-periodic} and Theorem~\ref{thm:regional-periodic}.

\begin{proof}
  In order to control the periodic components, we again need a notion of \emph{residue}. For a given integer $L$, we define its residue with respect to $\mathcal M$ to be the indexed list
  \[
    \Res_{\mathcal M}(L) = ((L, \hdots, L) \mod J_k : k \in S).
  \]
  Since $\mathbb Z^D/J_k$ is a finite quotient for all $k \in S$, the set of all possible residues is finite. Thus, there exists a residue $\hat{r} := (\hat{r}_{k}^d : k \in S, d \in D)$ such that $\hat{r} = \Res_{\mathcal M}(L)$ for infinitely many $L$ such that $\SIMPLEX_{\Part, L, \mathcal M} \in \Pol(\Gamma)$.

  As in the regional-periodic case, we apply the ring-theoretic Chinese Remainder Theorem to our quotient rings. Let $J = \bigcap_{k \in S} J_k$ be the intersection of the ideals of $\mathbb Z^D$. Note that $J$ is also an ideal of $\mathbb Z^D$. Furthermore, $\mathbb Z^D/J$ is finite, as any two elements of $x, y \in \mathbb Z^D$ with the same residue satisfy $x - y \in J_k$ for all $k \in S$, so $x - y \in J$, implying a finite number of cosets. This implies that we can identify $\hat{r}$, with an element of $\mathbb Z^D/J$.

  Now, let $\mathcal A = A^D$ and $\mathcal R = \mathbb Z^D/J$. Let $\mathcal R'$ be the subring of $\mathbb Z^D/J$ generated by $(1, \hdots, 1)$. It is not hard to see that $\hat{r} \in \mathcal R'$.  Recall that $\LPLE_{\mathcal A, \mathcal R' \subset \mathcal R}$ is the direct product
  \[
    \LPLE_{\mathcal A, \mathcal R'\subset \mathcal R} = \LP_{A^D} \times \LE_{\mathcal R' \subset \mathbb Z^D/J}.
  \]
  We claim that $\Gamma$ sandwiches $\LPLE_{\mathcal A, \mathcal R}$ via the maps $(g, h)$ where
  \begin{align}
    g(d) &= (e^d, \hat{r}e^d), d \in D\nonumber\\
    h((x_1, \hdots, x_{|D|}), r) &= M_k(r)\text{ where $k = \Part(x_1, \hdots, x_{|D|})$.}\label{eq:hh}
  \end{align}

  Since  $A$ is LP-solvable, and $\mathcal R' \subset \mathbb Z^D/J$ is a pair of commutative rings (and so is LE-solvable), by Theorem~\ref{thm:basiclp+affine}, it suffices to demonstrate  that $\Gamma$ sandwiches the following predicate
  \[
    \LPLE_{\Gamma} := \{R_i := (\Conv_{A^D}(g_1(P_i)), \Aff_{\mathcal R'}(g_{2}(P_i))): P_i \in \Gamma_P\}.
  \]
  As before, by definition, $g$ is a homomorphism from $\Gamma_P$ to $\LPLE_{\Gamma}$. Thus, it suffices to show
  \[\text{for all }(V, W) \in R_i,\text{ we have that }h(V, W) \in Q_i.\]
  
  Let $X^1, \hdots, X^m$ be the elements of $P_i$. Since $V \in \Conv_{A^D}(g_1(P_i))$, we have that there exists $\alpha_{1}, \hdots, \alpha_{m} \in [0, 1]$ summing to $1$ such that 
  \[
    V = \alpha_1g_{1}(X^1) + \cdots + \alpha_{m}g_1(X^m).
  \]
  
  Likewise, since $W \in \Aff_{\mathcal R'}(g_2(P_i))$, we have that there exist $r_1, \hdots, r_m \in \mathcal R'$ which sum to the identity $(1, \hdots, 1) \in \mathbb Z^D/J$ such that
  \[
    W = r_1g_{2}(X^1) + \cdots + r_mg_{2}(X^m).
  \]
  
  Consider $L$ sufficiently large (to be specified later) such that $\Res_{\mathcal M}(L) = \hat{r} \in R'$. Find nonnegative integer weights $w_{1}, \hdots, w_{m}$ satisfying the following conditions:
  \begin{align*}
    \sum_{k=1}^m w_{k} = L&&\text{(cardinality condition)}\\
    (w_k, \hdots, w_k) \in r_k\hat{r} + J&\text{ for all $k \in [m]$}&\text{(coset condition)}\\
    \left|w_{k} - \alpha_{k}L\right| \le 2|\mathbb Z^D/J|m&\text{ for all $k \in [m].$} &\text{(approximation condition)}
  \end{align*}
  This can be done by first estimating $\hat{w}_{k} = \lfloor \alpha_{k}L\rfloor$ and then adjusting each as little as possible (at most $|\mathbb Z^D/J|$)  so that $w_k$ are in the appropriate cosets. Then, it is not hard to check by the compatibility of the conditions that
    \[
      T := (L, \hdots, L) - \sum_{k=1}^m (w_k, \hdots, w_k) \in J
    \]
    If we then add the repeated element of $T$ to $w_1$, then the cardinality and coset constraints are satisfied.  Note that the entries of $T$ are bounded by $|\mathbb Z^D/J|m$, so the approximation condition is also still satisfied.

  Now consider
  \begin{align*}
    Y := \SIMPLEX_{\Part, L, \mathcal M}(&\underbrace{X^1, \hdots, X^1}_{w_{1}\text{ copies}}, \hdots, \underbrace{X^m, \hdots, X^m}_{w_{m}\text{ copies}}) \in Q_i.
  \end{align*}

  We seek to prove that $h(V, W) = Y$, showing that $h(V, W) \in Q_i$. For each coordinate  $k \in \{1, \hdots, \ar_i\}$ (recall $\ar_i$ is the arity of $P_i$, $Q_i$ and $R_i$) and $d \in D$ we can define
  \begin{align*}
    s_k^d &:= \frac{1}{L}\HamVec(\underbrace{X_k^1, \hdots, X_k^1}_{w_{1}\text{ copies}}, \hdots, \underbrace{X_k^m, \hdots, X_k^m}_{w_m\text{ copies}})_d.\\
              &= \frac{1}{L}\sum_{\beta = 1}^m w_{\beta} 1[X_k^\beta = d]\\
              &\approx \sum_{\beta = 1}^m \alpha_{\beta} 1[X_k^\beta = d] = V_{d, k},
  \end{align*}
  where $\approx$ means $O(\frac{1}{L})$ error. Thus, if $L$ is sufficiently large, for all $k \in [\ar_i]$, $(s^1_k, \hdots, s^{|D|}_k)$ will be in the same region of $\Part$ as $(V_{1,k}, \hdots, V_{|D|,k})$ because the regions are relatively open.

  Furthermore, for all $k \in [\ar_i]$ define
  \begin{align*}
    s^{\mathcal R}_k &= \sum_{\beta=1}^m (w_{\beta}, \hdots, w_{\beta})e^{X^\beta_k}\\
                     &\in J + \sum_{\beta=1}^m r_\beta(\hat r e^{X^\beta_k})\\
                     &= J + \sum_{\beta =1}^m r_\beta g_{b+1}(X^\beta_k) = W_k.
  \end{align*}

  Thus, for all $k \in [\ar_i]$,
  \begin{align*}
    Y_k &= M_{\overline{\Part}(s^{1}_k, \hdots, s^{|D|}_k)}(s^{\mathcal R}_k)\\
        &= M_{\Part(V_{1,k}, \hdots, V_{|D|,k})}(W_k)\text{ (by above discussion)}\\
        &= h_k(V, W)\text{ (by (\ref{eq:hh}))}.
  \end{align*}
  Thus, $Y = h(V, W)$, so we have established the homomorphic sandwich, completing the proof.
\end{proof}

\section{Conclusions}\label{sec:conclusion}

Our algorithms show how rich and diverse algorithms can be for promise CSPs as in comparison to classical CSP theory. In particular, finite promise CSPs can often demand algorithms which require infinite domains! There are many challenges for extending these algorithmic results to wider classes of polymorphisms. These challenges range from more topological inquiries to fundamental questions about infinite-domain CSPs.

One aspect of promise CSPs that was not utilized in this paper is that when the template $\Gamma$ is finite, $\Pol(\Gamma)$ is ``finitizable'' (c.f., \cite{DBLP:conf/soda/BrakensiekG18}), which means that there exists a constant $R_{\Gamma} > 0$, such that $f \in \Pol(\Gamma)$ if and only if all of its projections of arity $R_{\Gamma}$ are in $\Pol(\Gamma)$. Such a property may give a topological foothold (e.g., compactness) which could allow for more general classification. For instance, it is certainly possible that if a $\Gamma$ is finite, and $\Pol(\Gamma)$ contains (block) symmetric polynomials for arbitrarily large arities, then $\Pol(\Gamma)$ contains an infinite family of regional or regional-periodic polymorphisms (or some slight variant thereof) with consistent parameters. To prove such a result, a topological theory of polymorphisms needs to be developed.

Another important question is whether generalizations of the Basic LP, such as the Sherali-Adams or Sum-of-Squares hierarchies, correspond to classes of infinite CSPs that can be sandwiched by finite Promise-CSPs. Semidefinite programming may be especially useful for non-Boolean domains, as there is an algorithm known for Example 7 of Section~\ref{subsec:PCSP}, using SDPs [folklore].

Probably the most important--and daunting--question regarding promise CSPs is identifying the polymorphic dividing line, if it exists at all, between tractable and intractable promise CSPs. In the case of CSPs, having a single nontrivial cyclic polymorphism--$f(x_1, \hdots, x_n) = f(x_2, \hdots, x_n, x_1)$--is enough to imply tractability \cite{DBLP:conf/dagstuhl/BartoKW17}. One plausible conjecture for promise CSPs consistent with current knowledge is having infinitely many \emph{block transitive} polymorphisms suffices for tractability. that is $f : D^{B_1} \times \cdots \times D^{B_d} \to E$ which have the property that for all $i, j \in B_k$ for every $k$ there is a permutation $\pi$ of the coordinates such that $\pi(i) = j$. Even so, it is rather likely that there are tractable promise CSPs without an infinite family of polymorphisms of this form.

Note that there is still much work that needs to be done on the hardness side of the dichotomy. As shown by the struggles of the hardness of approximation community to solve the approximate graph coloring problem, stronger versions of the PCP theorem are desired. The recent breakthrough on the 2--to--2 conjecture \cite{DBLP:journals/eccc/DinurKKMS16,DBLP:conf/stoc/KhotMS17,DBLP:journals/eccc/DinurKKMS17,DBLP:journals/eccc/KhotMS18} is an encouraging step in this direction, although its impact on promise CSPs such as the approximate graph coloring problem is limited due to the fact that the current version lacks perfect completeness.

Another exciting direction for future exploration is understanding, for both CSPs and promise CSPs, what insight these polymorphisms shed on the existence of "fast" exponential-time algorithms for $\mathsf{NP}$-hard constraint templates. Such questions have been investigated for CSPs, most notably the work of \cite{DBLP:conf/soda/LagerkvistJNZ13}, which showed that the fundamental universal algebraic object are \emph{partial polymorphisms}, maps $f : D^L \to D \cup \{\perp\}$ for which the tuples mapping to $\perp$ are ignored. They also identified the ``easiest'' $\mathsf{NP}$-hard templates, but indicated that an exhaustive classification is currently out of reach. The perspective given in this work of considering threshold-periodic and regional-periodic polymorphisms can be easily extended to partial polymorphisms by adding $\perp$ as an extra element of the domain. The study of these families of partial polymorphisms and their utility in designing algorithms beating brute force is the subject of work in preparation.

\section*{Acknowledgments}
The authors thank Anupam Gupta and Ryan O'Donnell for helpful discussions about linear programming.

The authors also thank Libor Barto, Andrei Krokhin and Jakub Opr{\v s}al for a myriad of helpful comments on this paper at Dagstuhl Seminar 18231 on "The Constraint Satisfaction Problem: Complexity and Approximability."

\appendix

\section{Proofs of the Promise Homomorphism Theorems}\label{app:hom-proofs}

\basiclp*

\begin{proof}[Proof of Theorem~\ref{thm:basiclp}]
  We give an algorithm for both the decision and search version.
  \begin{framed}
    \begin{itemize}
    \item Write the Basic LP relaxation of $\Psi_P(x_1, \hdots, x_n)$.
    \item Solve the Basic LP over the ring $A$ to get a solution $(v_i \in A^k)_{i \in [n]}$. \textbf{Reject} if no solution.
    \item For all $i \in [n]$, set $y_i := h(v_i).$ \textbf{Accept} and output $(y_1, \hdots, y_n).$
    \end{itemize}
    \centering{\textbf{Algorithm \thesection.1.} Solving and rounding a Basic LP.}
  \end{framed}

  First we explain why this algorithm is correct. Assume $\Psi_P$ has a satisfying assignment, then the Basic LP must also have a satisfying assignment. Let $\LP_{\Gamma} = \{R_i := \Conv_{A^k}(S_i) : i \in I, S_i \subset \mathbb Z^{k\ar_i}\}$. Since $g$ is a homomorphism from $\Gamma_P$ to $\LP_{\Gamma}$, we have that $g(P_i) \subset R_i$ for all $i \in I$. In particular, this implies that $\Conv_{A^k}(g(P_i)) \subset R_i$. Thus, any solution to the Basic LP is a satisfying assignment of
  \[
    \Psi_R(x_1, \hdots, x_n) := \bigwedge_{j \in J} R_{i_j}(x_{j_1}, \hdots, x_{j_{\ar_{i_j}}}).
  \]
  Now, since $h$ is a homomorphism is a from $\LP_{\Gamma}$ to $\Gamma_Q$, any satisfying assignment to $\Psi_R$ (and thus to the Basic LP) maps via $h$ to a satisfying assignment to $\Psi_Q$. Thus, the algorithm correctly solves the search problem, and thus it also solves the decision problem.

  Finally, we explain why this algorithm lies in $\mathsf{P}^h$. Note that that Basic LP can be computed in linear time in the size of $\Psi_P$, and thus the instance can be solved in polynomial time since $A$ is LP-solvable (note that we need the range of $g$ to be in $\mathbb Z^k$ to ensure that the LP has integer coefficients). The ``rounding'' step uses an oracle to $h$, so $\PCSP(\Gamma) \in \mathsf{P}^h$.
\end{proof}

\affine*

\begin{proof}[Proof of Theorem~\ref{thm:affine}]
  Consider the following algorithm.
  \begin{framed}
    \begin{itemize}
    \item Write the affine relaxation of $\Psi_P(x_1, \hdots, x_n)$.
    \item Solve the affine relaxation over the pair $(R', R)$ to get a solution $r_1, \hdots, r_n \in R$. \textbf{Reject} if no solution.
    \item For all $i \in [n]$, set $y_i := h(r_i).$ \textbf{Accept} and output $(y_1, \hdots, y_n).$
    \end{itemize}
        \centering{\textbf{Algorithm \thesection.2.} Solving and rounding an affine relaxation.}
  \end{framed}

  First we explain why this algorithm is correct. Assume $\Psi_P$ has a satisfying assignment, then the affine relaxation must also have a satisfying assignment. Let $\LE_{\Gamma} = \{R_i := \Aff_{R'}S_i) : i \in I, S_i \subset R^{\ar_i}\}$. Since $g$ is a homomorphism from $\Gamma_P$ to $\LE_{\Gamma}$, we have that $g(P_i) \subset R_i$ for all $i \in I$. In particular, this implies that $\Aff_R(g(P_i)) \subset R_i$. Thus, any solution to the Basic LP is a satisfying assignment of
  \[
    \Psi_R(x_1, \hdots, x_n) := \bigwedge_{j \in J} R_{i_j}(x_{j_1}, \hdots, x_{j_{\ar_{i_j}}}).
  \]
  Now, since $h$ is a homomorphism is a from $\LE_{\Gamma}$ to $\Gamma_Q$, any satisfying assignment to $\Psi_R$ (and thus to the affine relaxation) maps via $h$ to a satisfying assignment to $\Psi_Q$. Thus, the algorithm correctly solves the search problem, and thus it also solves the decision problem.

  Like in the previous proof, the relaxation has size linear in the input. Since $(R', R)$ is LE-solvable, the relaxation can be solved in polynomial time. The last step uses an oracle to $h$, so $\PCSP(\Gamma) \in \mathsf{P}^h$.
\end{proof}

\basiclpplusaffine*

\begin{proof}[Proof of Theorem~\ref{thm:basiclp+affine}]
  We use an algorithm which is a combination of the techniques in Theorem~\ref{thm:basiclp} and Theorem~\ref{thm:affine}.
  \begin{framed}
    \begin{itemize}
    \item For each $A_j \in \mathcal A$
      \begin{itemize}
      \item Write the Basic LP relaxation of $\Psi_P(x_1, \hdots, x_n)$.
      \item Solve the Basic LP over the ring $A_j$ to get a solution $(v_{j,i} \in A_j^{k_j})_{i \in [n]}$. \textbf{Reject} if no solution.
      \end{itemize}
    \item For each $R'_j \subset R_j \in \mathcal R$
      \begin{itemize}
      \item Write the affine relaxation of $\Psi_P(x_1, \hdots, x_n)$.
            \item Solve the affine relaxation over the pair $R'_j \subset R_j$ to get a solution $r_{j,1}, \hdots, r_{j,n} \in R_j$. \textbf{Reject} if no solution.
      \end{itemize}
    \item For all $i \in [n]$, set $y_i := h(v_{1,i}, \hdots, v_{\ell,i}, r_{1,i}, \hdots, r_{m,i}).$ \textbf{Accept} and output $(y_1, \hdots, y_n).$
    \end{itemize}
        \centering{\textbf{Algorithm \thesection.3.} Solving multiple Basic LPs and affine relaxations with simultaneous rounding.}
  \end{framed}

  Let $\LPLE_{\Gamma} = \{R_i : i \in I\}$ be the particular CSP with signature the same signature as $\Gamma$ such that $g$ is a homomorphism from $\Gamma_P$ to $\LPLE_{\Gamma}$ and $h$ is a homomorphism from $\LPLE_{\Gamma}$ to $\Gamma_Q$. Assume that $\Psi_P$ has a satisfying assignment, then each Basic LP and affine relaxation is satisfiable. Then, by the same logic as the previous two proofs, the solutions to all the linear programs and linear systems put together satisfies the corresponding instance of $\LPLE_{\Gamma}$:
  \[
    \Psi_R(x_1, \hdots, x_n) := \bigwedge_{j \in J} R_{i_j}(x_{j_1}, \hdots, x_{j_{\ar_{i_j}}}).
  \]
  Finally, since $h$ is a homomorphism from $\LPLE_{\Gamma}$ to $\Gamma_Q$, any satisfying assignment to $\Psi_R$ maps to a satisfying assignment to $\Psi_Q$, so the algorithm is correct for the search version and thus also for the decision version.
  
  Like in the previous proof, the relaxation has size linear in the input. Since each $A_i$ is LP-solvable (and the map $g$ ensures the original LPs have integer coefficients) and each $R'_i\subset R_i$ is LE-solvable, the relaxation can be solved in polynomial time. The last step uses an oracle to $h$, so $\PCSP(\Gamma) \in \mathsf{P}^h$.
\end{proof}

\bibliographystyle{alpha}
\bibliography{../../bib/master_clean}

\end{document}